\documentclass[preprint,3p]{elsarticle1}

\usepackage{graphicx}

\usepackage{amssymb}
\usepackage{comment}

\usepackage{lineno}

\usepackage{amsmath}
\usepackage{amsthm}
\usepackage{epsfig}

\usepackage{mathrsfs}

\usepackage{psfrag}
\usepackage{color}

\journal{Linear Algebra and its Applications}

\newtheorem{theorem}{Theorem}
\newtheorem{lemma}[theorem]{Lemma}
\newtheorem{example}[theorem]{Example}
\newtheorem{definition}[theorem]{Definition}
\newtheorem{claim}{Claim}

\newtheorem{corollary}[theorem]{Corollary}
\newtheorem{remark}{Remark}
\newtheorem{question}{Question}
\newtheoremstyle{algstyle}%
  {10mm}       
  {10mm}       
  {\tt}   
  {0pt}        
  {\bfseries}  
  {\newline}   
  {10mm}       
  {\thmname{#1}\thmnumber{ #2}\thmnote{ (#3)}}          
\theoremstyle{algstyle}
\newtheorem{algorithm}{Algorithm}

\newtheoremstyle{algdashstyle}%
  {10mm}       
  {10mm}       
  {\tt}   
  {0pt}        
  {\bfseries}  
  {\newline}   
  {10mm}       
  {\thmname{#1}\thmnumber{ #2}$'$\thmnote{ (#3)}}          
\theoremstyle{algdashstyle}

\newcommand{\nw}[1]{%
\textbf{#1}%
}

\newcommand{\mnw}[1]{%
\boldsymbol{#1}%
}

\newcommand{\bbmatrix}[1]{%
\begin{bmatrix} #1 \end{bmatrix}%
}

\newcommand{\ppmatrix}[1]{%
\begin{pmatrix} #1 \end{pmatrix}%
}

\newcommand{\mydot}[1]{%
\stackrel{\text{\Large .}}{#1}%
}

\newcommand{\lrar}{\leftrightarrow}
\newcommand{\equivd}{:\equiv}

\newcommand{\subseteqn}{\hspace{0.1cm}\subseteq \hspace{0.1cm}}
\newcommand{\supseteqn}{\hspace{0.1cm}\supseteq \hspace{0.1cm}}
\newcommand{\equaln}{\hspace{0.1cm} = \hspace{0.1cm}}
\newcommand{\plusn}{\hspace{0.1cm} + \hspace{0.1cm}}
\newcommand{\inn}{\hspace{0.1cm} \in \hspace{0.1cm}}

\newcommand{\equivn}{\hspace{0.1cm} \equiv  \hspace{0.1cm}}
\newcommand{\lrarn}{\hspace{0.1cm} \lrar  \hspace{0.1cm}}

\newcommand{\V}{\mbox{$\cal V$}} 
\newcommand{\lnew}{\mbox{$ L$}} 
\newcommand{\F}{\mbox{$\cal F$}} 
\newcommand{\Vm}{{\mathcal V}_{M}}            			
\newcommand{\Va}{{\cal V}_{A}}              			

\newcommand{\Vabt}{{\cal V}_{AB}^T}
\newcommand{\Vbct}{{\cal V}_{BC}^T}
\newcommand{\transp}{{^T}}
\newcommand{\Vap}{{\cal V}_{AP}}            

    \newcommand{\0}{{\mathbf 0}}        
        
 \newcommand{\Vab}{{\cal V}_{AB}}   			
 \newcommand{\Vcd}{{\cal V}_{CD}}   			
 \newcommand{\Vaadash}{{\cal V}_{AA'}}   			
 \newcommand{\Vadashbdash}{{\cal V}_{A'B'}}   			
\newcommand{\Vhab}{\hat{{\cal V}}_{AB}}   			
   
\newcommand{\Vhatwodashbtwodash}{\hat{{\cal V}}_{A"B"}}   
 
\newcommand{\Vs}{{\cal V}_{S}}             
\newcommand{\Iaa}{{ I}_{AA'}}           			
\newcommand{\Iww}{{ I}_{WW'}}           			
\newcommand{\Iwminusw}{{ I}_{\dw(-\mydot{W'})}}           			
\newcommand{\Iwdw}{{I}_{W\mydot{W}}}  			
\newcommand{\Vsp}{{\cal V}_{SP}}           			
\newcommand{\Ksp}{{\cal K}_{SP}}           			
\newcommand{\Asp}{{\cal A}_{SP}}           			
\newcommand{\A}{{\cal A}}           			
\newcommand{\Kt}{{\cal K}_{T}}           			
\newcommand{\Ks}{{\cal K}_{S}}           			
\newcommand{\Kpq}{{\cal K}_{PQ}}           			
\newcommand{\Apq}{{\cal A}_{PQ}}           			
\newcommand{\Ksq}{{\cal K}_{SQ}}           			
\newcommand{\Asq}{{\cal A}_{SQ}}           			
\newcommand{\Q}{{\mathbb{Q}}}
\newcommand{\Vsq}{{\cal V}_{SQ}}    
    
\newcommand{\Vp}{{\cal V}_{P}}              
\newcommand{\Vbc}{{\cal V}_{BC}}              
     
\newcommand{\Vabperp}{{\cal V}_{AB}^{\perp}}  
 
\newcommand{\Vpq}{{\cal V}_{PQ}}            			
\newcommand{\Vb}{{\cal V}_{B}}              			
\newcommand{\Vwdwm}{{\cal V}_{W\mydot{W}M}}  			
\newcommand{\Vadjwdw}{{\cal V}^a_{W'\mydot{W'}}}  			
\newcommand{\Vpdwmu}{{\cal V}_{P\mydot{W}M_u}}  			
\newcommand{\Xwdw}{{\cal X}_{W\mydot{W}}}  			
\newcommand{\Ywdw}{{\cal Y}_{W\mydot{W}}}  			
\newcommand{\Vdwmy}{{\cal V}_{\mydot{W}M_y}}  			
\newcommand{\Vtdwmy}{{\tilde{\cal V}}_{\mydot{W'}M_y'}}  			
\newcommand{\Vtdpmy}{{\tilde{\cal V}}_{\mydot{P'}M_y'}}  			
\newcommand{\Vwdwmumy}{{\cal V}_{W\mydot{W}M_uM_y}}  			
\newcommand{\Vwdwmu}{{\cal V}_{W\mydot{W}M_u}}  			
\newcommand{\Vdpmy}{{\cal V}_{\mydot{P}M_y}}  			
\newcommand{\Vpmu}{{\cal V}_{PM_u}}  			
\newcommand{\Vpdpmumy}{{\cal V}_{P\mydot{P}M_uM_y}}  			
\newcommand{\Vpdpmu}{{\cal V}_{P\mydot{P}M_u}}  			
\newcommand{\Vtwdw}{{\tilde{\cal V}}_{W'\mydot{W'}}}  			
\newcommand{\Vtwdwm}{{\tilde{\cal V}}_{W'\mydot{W'}M_u'M_y'}}  			
\newcommand{\Vtpdp}{{\tilde{\cal V}}_{P'\mydot{P'}}}  			
\newcommand{\Vtpdpm}{{\tilde{\cal V}}_{P'\mydot{P'}M_u'M_y'}}  			
\newcommand{\Vwmu}{{\cal V}_{WM_u}}  			
\newcommand{\Vmydw}{{\cal V}_{\mydot{W}M_y}}  			
\newcommand{\Vdwdp}{{\cal V}_{\mydot{W}\mydot{P}}}  			
\newcommand{\Vonedwdp}{{\cal V}^1_{\mydot{W}\mydot{P}}}  			
\newcommand{\Vonetildedwdp}{{\cal V}^1_{\mydot{\tilde{W}}\mydot{{P}}}}  			
\newcommand{\Vtwotildedwdp}{{\cal V}^2_{\mydot{\tilde{W}}\mydot{{P}}}}  			
\newcommand{\tildedwdp}{{\mydot{\tilde{W}}\mydot{{P}}}}  			
\newcommand{\wtildedw}{{W\mydot{\tilde{W}}}}  			
\newcommand{\tildedw}{{\mydot{\tilde{W}}}}  			
\newcommand{\Vtwodwdp}{{\cal V}^2_{\mydot{W}\mydot{P}}}  			
\newcommand{\Vtwowp}{{\cal V}^2_{{W}{P}}}  			
\newcommand{\Vtwodwdq}{{\cal V}^2_{\mydot{W}\mydot{Q}}}  			
\newcommand{\Vtwodpdq}{{\cal V}^2_{\mydot{P}\mydot{Q}}}  			
\newcommand{\Vwp}{{\cal V}_{WP}}  			
\newcommand{\Vonewp}{{\cal V}^1_{WP}}  			
\newcommand{\Vonepq}{{\cal V}^1_{PQ}}  			
\newcommand{\Vonewq}{{\cal V}^1_{WQ}}  			
\newcommand{\Vtonewp}{{\tilde{\cal V}}^1_{W'P'}}  			
\newcommand{\Vttwodwdp}{{\tilde{\cal V}}^2_{\dwd\dPd}}  			
\newcommand{\Vw}{{\cal V}_{W}}  			
\newcommand{\Vtildedw}{{\cal V}_{\mydot{\tilde{W}}}}  			
\newcommand{\Vwdw}{{\cal V}_{W\mydot{W}}}  			
\newcommand{\Vonewdw}{{\cal V}^u_{W\mydot{W}}}  			
\newcommand{\oVonewdw}{{\cal V}^1_{W\mydot{W}}}  			
\newcommand{\Vtwowdw}{{\cal V}^l_{W\mydot{W}}}  			
\newcommand{\oVtwowdw}{{\cal V}^2_{W\mydot{W}}}  			
\newcommand{\wdw}{{W\mydot{W}}}  			
\newcommand{\Vwdwk}{({\cal V}_{W\mydot{W}})^{(k)}}  			
\newcommand{\Vdw}{{\cal V}_{\mydot{W}}}  			
\newcommand{\Vpdpm}{{\cal V}_{P\mydot{P}M}}  			
\newcommand{\Vqdqm}{{\cal V}_{Q\mydot{Q}M}}  			
\newcommand{\Vpdp}{{\cal V}_{P\mydot{P}}}  			
\newcommand{\Vqdq}{{\cal V}_{Q\mydot{Q}}}  			
\newcommand{\Vpdpk}{({\cal V}_{P\mydot{P}})^{(k)}}  			
\newcommand{\dw}{{\mydot{W}}}  			
\newcommand{\dws}{{\mydot{w}}}  			
\newcommand{\dwd}{{\mydot{W'}}}  			
\newcommand{\dW}{{\mydot{w}}}  			
\newcommand{\dP}{{\mydot{P}}}  			
\newcommand{\dPd}{{\mydot{P'}}}  			
\newcommand{\dQ}{{\mydot{Q}}}  			
\newcommand{\dwdp}{{\mydot{W}\mydot{P}}}  			
\newcommand{\dotp}{{\mydot{P}}}  			
\newcommand{\dwsmall}{{\mydot{w}}}  			

\newcommand{\T}[0]{{\cal T}}                    	
\newcommand{\E}{\mbox{$\cal E$}} 
 
\newcommand{\G}[0]{{\cal G}}                       
  
\newcommand{\K}[0]{{\cal K}}                       

\newcommand{\KSP}{\mbox{${\cal K}_{SP}$}}    		
\newcommand{\KSQ}{\mbox{${\cal K}_{SQ}$}}    		
\newcommand{\KPQ}{\mbox{${\cal K}_{PQ}$}}    		
\newcommand{\M}{\mbox{$\cal M$}}

\newcommand{\N}[0]{{\cal N}}    							

\newcommand{\W}[0]{{\cal W}}                       
 
\newcommand{\B}{\mbox{${\cal B}$}}  				

\newcommand{\f}{\mbox{${\bf f}$}}        				

\newcommand{\g}{\mbox{${\bf g}$}}        				

\newcommand{\pa}{\mbox{${+}_{a}$}}      				
\newcommand{\pb}{\mbox{${+}_{b}$}}      				
\newcommand{\pc}{\mbox{${+}_{c}$}}      				
\newcommand{\pdw}{\mbox{${+}_{\mydot{w}}$}}      				
\newcommand{\pdwd}{\mbox{${+}_{\mydot{w}'}$}}      				
\newcommand{\pdp}{\mbox{${+}_{\mydot{p}}$}}      				

\newcommand{\x}{\mbox{${\bf x}$}}             		
\newcommand{\y}{\mbox{${\bf y}$}}             		

\newcommand{\Vy}[1]{{\cal V} _Y #1}
\newcommand{\VX}[1]{{\cal V} _X #1}

\newcommand{\s}[1]{ short  \ #1}                 

\newcommand{\al}{\Box\,}

\newcommand{\ldw}{\lambda ^{\mydot{w}}}
\newcommand{\ldws}{\lambda ^{\mydot{w}}}
\newcommand{\ldwd}{\lambda ^{\mydot{w}'}}

\newcommand{\ldp}{\lambda ^{\mydot{p}}}

\begin{document}

\begin{frontmatter}



\title{Implicit Linear Algebra and Basic Circuit Theory
}

\author[hn]{H. Narayanan\corref{cor1}}
\ead{hn@ee.iitb.ac.in}
\cortext[cor1]{Corresponding author}
\author[hari]{Hariharan Narayanan}
\ead{hariharan.narayanan@tifr.res.in}
\address[hn]{Department of Electrical Engineering, Indian Institute of Technology Bombay}
\address[hari]{School of Technology and Computer Science, Tata Institute of Fundamental Research}

\begin{abstract}
In this paper we derive some basic results of circuit theory 
using `Implicit Linear Algebra' (ILA). This approach has the advantage of 
simplicity and generality. Implicit linear algebra is outlined in
\cite{HNarayanan1986a, narayanan1987topological, HNarayanan1997,HNarayanan2009}.

A vector space  $\Vs$ is a collection of row vectors, over a field $\mathbb{F}$ on a finite column set $S,$ closed under addition and multiplication by elements 
of $\mathbb{F}.$ We denote the space of all vectors on $S$ by $\F_S$ 
and the space containing only the zero vector on $S$ by $\0_S.$
The dual $\Vs^{\perp}$ of a vector space $\Vs$ is the collection
of all vectors whose dot product with vectors in $\V_S$ is zero.

The basic operation of ILA is a linking operation ('matched composition`) between vector spaces $\V_{SP},\V_{PQ}$
(regarded as collections of row vectors on column sets $S\cup P, P\cup Q,$
respectively with $S,P,Q$ disjoint)
defined by
$\V_{SP}\leftrightarrow \V_{PQ}\equivd
\{(f_S,h_Q):((f_S,g_P)\in \V_{SP}, (g_P,h_Q) \in \V_{PQ}\},$
and another ('skewed composition`)
defined by
$\V_{SP}\rightleftharpoons \V_{PQ}\equivd
\{(f_S,h_Q):((f_S,g_P)\in \V_{SP}, (-g_P,h_Q) \in \V_{PQ}\}.$
These operations can also be described in terms of sum, intersection,
contraction and restriction of vector spaces.
\\ 
%
%
The basic results of ILA are 
 the Implicit Inversion Theorem 
(which states that
$\V_{SP}\lrar(\V_{SP}\lrar \V_S)= \V_S,$ iff 
$\V_{SP}\lrar \0_P\subseteq \V_S\subseteq \V_{SP}\lrar\F_S$)  and Implicit Duality Theorem
(which states that $(\V_{SP}\leftrightarrow \V_{PQ})^{\perp}= (\V_{SP}^{\perp}\rightleftharpoons \V_{PQ}^{\perp}$).

We show that the operations and results of ILA
are useful in 
understanding basic circuit theory.
We illustrate this by using ILA to 
 present a generalization of Thevenin-Norton theorem
where we compute  multiport behaviour using
adjoint multiport termination through a gyrator
and a very general  version of maximum power transfer theorem,
which states that the port conditions that appear, during adjoint multiport termination through an ideal transformer, correspond to maximum power transfer.

\end{abstract}


\end{frontmatter}


\section{Introduction}
\label{sec:intro}
Implicit linear algebra (ILA) was originally developed to better exploit
the topology of the electrical network, 
while analysing it.
Two basic topological methods which use  ILA are 
`multiport decomposition' (see \cite{HNarayanan1986a}, Section 8 of \cite{HNarayanan1997,HNarayanan2009}),  which could be regarded as the most  natural to circuit theory, and the method of analysis using topological transformations (\cite{narayanan1987topological}), where the networks change their topology with additional linear constraints
but always with a guarantee that the new set of constraints are equivalent
to the original set of network constraints.
The first  half of the book
\cite{HNarayanan1997,HNarayanan2009} deals with these ideas.

A characteristic feature of the ILA approach is, as far as possible, to avoid
speaking of networks in terms of matrices but rather in terms of other derived networks. The variables  of interest would often be buried in a much larger set
of network variables and we have to deal with the former {\it implicitly}
without trying to {\it eliminate} the  latter.

When one connects linear subnetworks,  
there is no a priori {\it direction} to the flow of `information', merely {\it interaction,} which could be regarded as working both ways. 
To capture this situation, linear transformations are not suitable,
rather one needs linear relations. One is naturally led to the question whether
parts of linear algebra can be developed based on linear relations.
The answer is in the affirmative 
 and implicit linear algebra satisfies this requirement (\cite{narayanan2016}). 
In ILA, since we work with linear relations
 which link subspaces rather than with maps which link vectors, we  manipulate
subspaces rather than vectors.     
From the  point of view of the present paper,
it is more important that there are applications to circuit theory. 
The basic operation in ILA is that of matched (skewed) composition of indexed spaces
described below.
This corresponds precisely to the connection of multiports across
some of their ports.

An essential feature of linear circuits and systems theory is that 
there are  
 two kinds of constraints: one that is related to connection and 
the other, to device (or subsystem) characteristics. The former are 
combinatorial constraints involving coefficient matrices which have 
$0,\pm 1$ entries while the latter can be thought of as general 
linear algebraic ones with coefficient matrices which have rational or real entries. The usual way to exploit this situation is to 
use graph theory for writing the former constraint equations  and combine this with the 
latter using linear algebra. In this approach the final set of equations 
involve floating point  numbers and the influence of the combinatorial 
constraints is felt primarily in the sparsity of the final coefficient matrix.

A second approach which is useful for studying the solvability of the 
systems under consideration is through the use of matroid theory 
(specifically through matroid union and matroid intersection theorems 
\cite{edm65a,iritomizawa,iritomizawa2,irisurvey,irireview,irifujishige,iriapplications,iriprogress,murotairi1,murotabook0, murotabook,recski89}).
An approach that is close in spirit to ILA is that of behavioural system theory
 (\cite{willems1991paradigms,willems1997,willemstrentelman,van2004bisimulation,vanderscaftport}).
Here, again, systems are taken as they are presented in the governing constraints
 without first putting them in a form which contains only the variables of interest.
To illustrate the approach, in place of state and output linear equations of the kind
\begin{align}
\label{eqn:behaviour0}
\ppmatrix{\dot{x}\\y}=\ppmatrix{f_1(x,u)\\f_2(x,u)},
\end{align}
one works with equations of the kind 
\begin{align}
\label{eqn:behaviour1}
{f(w,\dot{z}, z,u,y)=0},
\end{align}
which may be taken to be the original constraints.
The variable entries of the vector $x$ of Equation \ref{eqn:behaviour0} can be taken to be contained 
in the variable entries of the vector $z$ of Equation \ref{eqn:behaviour1}.
After suitable elimination of variables we can get the former equation 
from the latter.
By working with the original constraints, instead of the derived ones of Equation \ref{eqn:behaviour0}, we are able to better exploit the structure of the system 
and avoid the externally imposed canonical structure of Equation \ref{eqn:behaviour0}. In this sense, this approach is similar to the ILA approach. However, this method does not have as one of its aims, the exploitation of  the difference in the nature of 
the connection and device characteristic constraints. Many of the proof 
techniques and algorithms work with vectors with polynomial entries rather
than with real or rational entries. 

The ILA approach, where also we can deal with the variables of interest 
implicitly, allows us to deal with combinatorial constraints more effectively,
while remaining within the domain of linear algebra of vectors over reals 
or rationals.
Indeed, one can develop a topological network theory in this approach (\cite{HNarayanan1986a,narayanan1987topological,HNarayanan1997}). Although we do not discuss dynamical systems in this paper,
the ILA approach can be used to derive the classical controllability, observability
theory (\cite{kalman,Wonham1978}) for dynamical systems (\cite{HNPS2013,narayanan2016}) and when such systems are based on electrical 
networks, in place of state and output equations, one can work with 
a network suitably decomposed into static and dynamic multiports 
(\cite{narayanan2016}). The decomposition is linear time to carry out 
and preserves the topological structure of the network in an essential way.

Basic circuit theory has many procedures which usually work, but whose failure
does not imply
that the task cannot be completed in another
way.
This inadequacy can be alleviated through implicit linear algebra.
In this paper, we ilustrate its use by developing
procedures  which are shown to work under general
conditions. For instance, through the use of ILA,
we give a generalization of Thevenin-Norton theorem,
where we compute  multiport behaviour using
adjoint multiport termination through a gyrator,
and a very general  version of maximum power transfer theorem,
where we use 
adjoint multiport termination through an ideal transformer.
Our result for the Thevenin-Norton theorem is  valid for `regular' multiports which are defined to be
linear multiports  with nonvoid set of solutions for arbitrary source 
values in the device characteristic and with unique interior solutions 
for  given port conditions. In the case of the maximum power transfer
theorem, our result is valid for general linear multiports.
We also prove duality properties
for general linear multiports using the implicit duality 
theorem described below.

Implicit Linear Algebra (ILA) allows us to avoid needless computations.
For instance, in the case of the maximum power transfer theorem,
the usual statement is that this occurs when we terminate
the multiport by the transpose (adjoint) of the Thevenin impedance.
This apparently requires the computation of the Thevenin impedance.
Implicit linear algebra allows us to restate this as `maximum power
transfer occurs corresponding to the port conditions that obtain
when the original multiport is terminated  by its adjoint.'
Now, building the adjoint of the multiport is usually much easier 
(in practice, linear time on the size of the network)
than computing its Thevenin impedance.
Of course, the network that results  by the termination with the adjoint
may have no solution. But this guarantees that maximum power transfer
is not possible. Such a guarantee is not available when the statement
is in terms of Thevenin impedance, because the latter might not exist
and yet maximum power transfer might be possible.

We now give a brief account of the basic ideas, which are stated in terms 
of vector spaces for simplicity.\\
A {\it {vector}} ${f}$ on $X$ over $\mathbb{F}$ is a function $f:X\rightarrow \mathbb{F}$ where $\mathbb{F}$ is a field. It would usually be represented as $f_X.$ A collection of such vectors closed under addition and scalar multiplication is a {\it vector space} denoted by $\V_X.$
When the vector space is on $X\cup Y, \ X,Y,$ disjoint, it is 
denoted by $\V_{XY}.$


The usual {\it sum} and {\it intersection} of vector spaces are given extended meanings as follows.
$$\V_X+\V_Y\equivd (\VX\oplus \ \0_{Y\setminus X})+ (\Vy\oplus \ \0_{X\setminus Y}),$$
$$\V_X\cap \V_Y\equivd (\VX\oplus \F_{Y\setminus X})\cap (\Vy\oplus \F_{X\setminus Y}),$$
where $\0_Z$ represents the space containing only the $0$ vector on $Z$ and $ \F_Z$ represents the collection of all vectors on $Z.$

Given $\V_S,T\subseteq S,$ the {\it restriction} of $\V_S$ to $T,$ denoted by
$\V_S\circ T,$ is the
collection of all $f_T,$ where $(f_T,f'_{S\setminus T}),$ for some $f'_{S\setminus T},$
belongs to $\V_S.$

The {\it contraction} of $\V_S$ to $T$  denoted by
$\V_S\times T,$  is the
collection of all $f_T,$ where $(f_T,0_{S\setminus T}),$
belongs to $\V_S.$

The {\it matched composition} of $\Vsp,\Vpq,$ with  $S,P,Q$ pairwise disjoint, is denoted by $\Vsp\lrar \Vpq$ and is defined to be the collection of all
$(f_S,h_Q)$ such that there exists some $g_P$ with \\
$(f_S,g_P)\in \Vsp,
(g_P,h_Q)\in \Vpq .
$

The
 {\it skewed composition} of $\Vsp,\Vpq,$ with $ S,P,Q$ pairwise disjoint,
 is denoted by $\Vsp\rightleftharpoons \Vpq$ and is defined to be the collection of all
$(f_S,h_Q)$ such that there exists some $g_P$ with \\
$(f_S,-g_P)\in \Vsp,
(g_P,h_Q)\in \Vpq.$

The {\it dot product} of $f_S, g_S $ is denoted by $ \langle f_S, g_S \rangle $ and is defined to be $\sum_{e\in S}f(e)g(e).$

The {\it complementary orthogonal space} to $\V_S$
is denoted by
${\V^{\perp}_S} $ and is defined to be\\
 $\{ g_S: \langle f_S, g_S \rangle = 0,\ f_S\in \V_S \}.$

The two basic results (\cite{HNarayanan1986a}) are as follows.

(IIT) The {\it implicit inversion theorem} states that the equation
$$ \Vsp\lrar \Vpq=\Vsq,$$ with  specified $ \Vsp, \Vsq,$ but $\Vpq$ as unknown,  has a solution iff $\Vsp\circ S\supseteq \Vsq \circ S,
\Vsp\times S\subseteq \Vsq \times S$\\
 and, further, under the additional conditions
$\Vsp\circ P\supseteq \Vpq \circ P,
\Vsp\times P\subseteq \Vpq \times P,$ it has a unique solution.

(IDT) The {\it implicit duality theorem} states that
$$(\Vsp\lrar \Vpq)^{\perp}=(\Vsp^{\perp}\rightleftharpoons \Vpq^{\perp}).$$
IIT could be regarded as a generalization of the usual existence-uniqueness result
for the solution of the equation $ Ax=b.$ In the context of this paper, it is very much more useful. IDT was a folklore result in electrical network theory
probably known informally to G.Kron (\cite{kron39},\cite{kron63}). An equivalent result  is stated
with a partial proof in \cite{belevitch68}.
It can be regarded as a generalization
of the result $(AB)^T=B^TA^T.$

We now give a brief outline of the paper.\\
Section \ref{sec:Preliminaries} is on preliminary definitions and results 
from linear algebra and graph theory.\\
Section \ref{sec:matched}
is on the operations of matched and skewed composition of collections of vectors
and their relation to  the connection of multiport behaviours across ports.
\\
Section \ref{sec:iit}
is on the first of the two basic results of ILA viz. Implicit Inversion Theorem (IIT)
 and
Implicit Duality Theorem (IDT), and 
on immediate applications.
IIT gives us, if  a linear multiport is connected to another and results in 
a third, conditions under which the port behaviour of one of the connected multiports can be
recovered when the other two are known.
\\
Section \ref{sec:idt}
is on the second of the two basic results of ILA viz. Implicit Duality Theorem (IDT). 
IDT tells us, for instance, that when the device characteristic is replaced by its  adjoint, the multiport behaviour also gets similarly replaced.
\\
Section \ref{sec:linearalgebrabehaviour}
is on the explicit computation of  solution and port behaviour of multiport networks. \\
Section \ref{subsec:regular}
characterizes regular mutiports (linear multiports which have nonvoid set of solutions 
for arbitrary source values in the device characteristic and unique 
interior solution corresponding to given port condition) using the 
ideas of the preceding section.\\
Section \ref{sec:computingbehaviour}
presents a generalization of Thevenin-Norton Theorem that is 
valid for all regular multiports. This uses the termination of the 
regular multiport by its adjoint through the affine version of a gyrator.
It is shown that such a termination always results in a network 
with a unique solution and, what is more, can be handled by 
currently available linear circuit simulators.
\\
Section \ref{sec:maxpower} derives the most general port conditions for 
stationarity of power transfer through the ports of a linear multiport 
and shows that these conditions are realized, if they can be realized at all, when the multiport 
is connected to its adjoint through an ideal transformer.
It is shown that for passive mutiports the stationarity of power transfer
corresponds to maximum power transfer and for strictly passive multiports
the maximum power transfer conditions are always statisfied.
\\
Section \ref{sec:conclusions} is on conclusions.
\\
The appendix contains proofs of general versions of IIT and IDT
and a brief description of the maximum power transfer for the complex case.
\section{Preliminaries}
\label{sec:Preliminaries}
The preliminary results and the notation used are from \cite{HNarayanan1997}.

A \nw{vector} $\mnw{f}$ on a finite set $X$ over $\mathbb{F}$ is a function $f:X\rightarrow \mathbb{F}$ where $\mathbb{F}$ is a field. 


%

The {\bf sets} on which vectors are defined are  always {\bf finite}. When a vector $x$ figures in an equation, we use the 
convention that $x$ denotes a column vector and $x^T$ denotes a row vector such as
in `$Ax=b,x^TA=b^T$'. Let $f_Y$ be a vector on $Y$ and let $X \subseteq Y$. The \textbf{restriction $f_Y|_X$} of $f_Y$ to $X$ is defined as follows:\\
$f_Y|_X \equivd g_X, \textrm{ where } g_X(e) = f_Y(e), e\in X.$


When $f$ is on $X$ over $\mathbb{F}$, $\lambda \in \mathbb{F},$ then  the \nw{scalar multiplication} $\mnw{\lambda f}$ of $f$ is on $X$ and is defined by $(\lambda f)(e) \equivd \lambda [f(e)]$, $e\in X$. When $f$ is on $X$ and $g$ on $Y$ and both are over $\mathbb{F}$, we define $\mnw{f+g}$ on $X\cup Y$ by \\
$(f+g)(e)\equivd f(e) + g(e),e\in X \cap Y,\ (f+g)(e)\equivd  f(e), e\in X \setminus Y,
\ (f+g)(e)\equivd g(e), e\in Y \setminus X.
$
(For ease in readability, we  use $X-Y$ in place of $X \setminus Y.$)

When $X, Y, $ are disjoint,  $f_X+g_Y$ is written as  $\mnw{(f_X, g_Y)}.$ When $f,g$ are on $X$ over $\mathbb{F},$ the \textbf{dot product} $\langle f, g \rangle$ of $f$ and $g$ is defined by 
$ \langle f,g \rangle \equivd \sum_{e\in X} f(e)g(e).$
When $X$, $Y$ are disjoint, $\mnw{X\uplus Y}$ denotes the disjoint
union of $X$ and $Y.$ A vector $f_{X\uplus  Y}$ on $X\uplus Y$ is  written as $\mnw{f_{XY}}.$

We say $f$, $g$ are \textbf{orthogonal} (orthogonal) iff $\langle f,g \rangle$ is zero.

An \nw{arbitrary  collection} of vectors on $X$ 
is denoted by $\mnw{\mathcal{K}_X}$. 
When $X$, $Y$ are disjoint we usually write $\mathcal{K}_{XY}$ in place of $\mathcal{K}_{X\uplus Y}$.
We write $\K_{XY}\equivd \K_X\oplus \K_Y$ iff
$\K_{XY}\equivd\{f_{XY}:f_{XY}=(f_X,g_Y), f_X\in \K_X, g_Y\in \K_Y\}.$
We refer to $\K_X\oplus \K_Y$ as the \nw{direct sum} of $\K_X, \K_Y.$ 

A collection $\K_X$ is a \nw{vector space} on $X$ iff it is closed under 
addition and scalar multiplication. 
The notation $\mnw{\V_X}$ always denotes
a vector space on $X.$
For any collection $\K_X,$  $\mnw{span(\K_X)}$ is the vector space of all
linear combinations of vectors in it.
We say $\A_X$ is an \nw{affine space} on $X,$ iff it can be expressed as
$x_X+\V_X,$ where $x_X$ is a vector and $\V_X,$ a vector space on $X.$
The latter is unique for $\A_X$ and is said to be its \nw{vector space translate}.

For a vector space  $\V_X,$ since we take $X$ to be finite,
any maximal independent subset of $\V_X$ has size less than or equal to $|X|$ and this 
size can be shown
to be unique. A maximal independent subset of a vector
space $\V_X$ is called its \nw{basis} and its  size 
is called the  {\bf dimension} or \nw{rank} of $\V_X$ and denoted by ${\mnw{dim}(\V_X)}$
 or by ${\mnw{r}(\V_X)}.$
For any collection of vectors $\K_X,$
the rank $\mnw{r}(\K_X)$
is defined to be $dim(span(\K_X)).$
The collection of all linear combinations of the rows of a matrix $A$ is a vector space 
that is denoted by $row(A).$

For any collection of vectors
$\mathcal{K}_X,$   the collection $\mnw{\mathcal{K}_X^{\perp}}$ is defined by
$ {\mathcal{K}_X^{\perp}} \equivd \{ g_X: \langle f_X, g_X \rangle =0\},$
It is clear that $\mathcal{K}_X^{\perp}$ is a vector space for 
any $\mathcal{K}_X.$ When $\mathcal{K}_X$ is a vector space $\V_X,$
 and the underlying set $X$ is finite, it can be shown that $({\mathcal{V}_X^{\perp}})^{\perp}= \mathcal{V}_X$ 
and  $\mathcal{V}_X,{\mathcal{V}_X^{\perp}}$ are said to be \nw{complementary orthogonal}. 
The symbol $0_X$ refers to the \nw{zero vector} on $X$ and $\mnw{0_X}$  refers to the \nw{zero vector space} on $X.$ The symbol $\mnw{\F_X}$  refers  to the collection of all vectors on $X$ over the field in question.
It is easily seen, when $X,Y$ are disjoint, and $\K_X, \K_Y$  
contain zero vectors, that $(\K_X\oplus \K_Y)^{\perp}=
\K_X^{\perp}\oplus\K_Y^{\perp}.$

A matrix of full row rank, whose rows generate a vector space $\V_X,$
is called a \nw{representative matrix} for $\V_X.$
A representative matrix which can be put in the form $(I\ |\ K)$ after column
permutation, is called a \nw{standard representative matrix}.
It is clear that every vector space has a  standard representative matrix.
If $(I\ |\ K)$ is  a standard representative matrix of $\V_X,$
it is easy to see  that $(-K^T|I)$ is a standard representative matrix of $\V^{\perp}_X.$
Therefore we must have
\begin{theorem}
\label{thm:perperp}
Let $\V_X$ be a vector space on $X.$ Then\\
$r(\V_X)+r(\V^{\perp}_X)=|X|$ and $((\V_X)^{\perp})^{\perp}=\V_X.$
\end{theorem}
\begin{remark}
\label{rem:realtocomplex}
When the field $\mathbb{F}\equivd \mathbb{C},$ it is usual to interpret
the dot product $\langle f_X, g_X \rangle$ of $f_X,g_X$ to be 
the inner product $\Sigma f(e)\overline{g(e)}, e\in X,$
where $\overline{g(e)}$ is the complex conjugate of $g(e).$
In place of $\V^{\perp}_X$ we must use $\V^{*}_X,$ 
where $\V^{*}_X\equivd \{ g_X: \langle f_X, g_X \rangle =0\},$
 $\langle f_X, g_X \rangle$ being taken to be as above.
The definition of adjoint, which is introduced later, must be in terms of 
$\V^{*}_X$ instead of in terms of $\V^{\perp}_X.$ 
\\
In this case, in place of the transpose of a matrix $Z$ we use $Z^*,$ the conjugate
transpose of $Z.$
\\
The above interpretation of dot product is {\it essential} if we wish to 
extend the development in Sections \ref{sec:computingbehaviour},
\ref{sec:maxpower} to the complex field. 
%
%
%
%
%
%
\\
The  
Implicit Duality Theorem (Theorem \ref{thm:idt0}), would go through
with either definition of dot product and the corresponding 
definition of orthogonality.
\end{remark}


The collection
$\{ (f_{X},\lambda f_Y) : (f_{X},f_Y)\in \mathcal{K}_{XY} \}$
is denoted by
$ \mnw{\mathcal{K}_{X(\lambda Y)}}.
$
When $\lambda = -1$ we would write $ {\mathcal{K}_{X(\lambda Y)}}$  more simply as $\mnw{\mathcal{K}_{X(-Y)}}.$
Observe that $(\mathcal{K}_{X(-Y)})_{X(-Y)}=\mathcal{K}_{XY}.$

We say sets $X$, $X'$ are \nw{copies of each other} iff they are disjoint and there is a bijection, usually clear from the context, mapping  $e\in X$ to $e'\in X'$.
When $X,X'$ are copies of each other, the vectors $f_X$ and $f_{X'}$ are said to be copies of each other with  $f_{X'}(e') \equivd  f_X(e), e \in X.$ 
The copy $\K_{X'}$ of $\K_X$ is defined by
 $\K_{X'}\equivd\{f_{X'}:f_X\in \K_X\}.$
When $X$ and $X'$ are copies of each other, the notation for interchanging the positions of variables with index sets $X$ and $X'$ in a collection $\mathcal{K}_{XX'Y}$ is given by $\mnw{(\mathcal{K}_{XX'Y})_{X'XY}}$, that is\\
$(\mathcal{K}_{XX'Y})_{X'XY}
 \equivd \{(g_X,f_{X'},h_Y)\ :\ (f_X,g_{X'},h_Y) \in \mathcal{K}_{XX'Y},\ g_X\textrm{ being copy of }g_{X'},\ f_{X'}\textrm{ being copy of }f_X  \}.$
An affine space $\mathcal{K}_{XX'}$ is said to be {\bf proper}
iff the rank of its vector space translate is $|X|=|X'|.$


\subsection{Sum and Intersection}
Let $\mathcal{K}_{SP}$, $\mathcal{K}_{PQ}$ be collections of vectors on sets $S\uplus P,$ $P\uplus Q,$ respectively, where $S,P,Q,$ are pairwise disjoint. The \nw{sum} $\mnw{\mathcal{K}_{SP}+\mathcal{K}_{PQ}}$ of $\mathcal{K}_{SP}$, $\mathcal{K}_{PQ}$ is defined over $S\uplus P\uplus Q,$ as follows:\\
 $\mathcal{K}_{SP} + \mathcal{K}_{PQ} \equivd  \{  (f_S,f_P,0_{Q}) + (0_{S},g_P,g_Q), \textrm{ where } (f_S,f_P)\in \mathcal{K}_{SP}, (g_P,g_Q)\in \mathcal{K}_{PQ} \}.$\\
Thus,
$\mathcal{K}_{SP} + \mathcal{K}_{PQ} \equivd (\mathcal{K}_{SP} \oplus \0_{Q}) + (\0_{S} \oplus \mathcal{K}_{PQ}).$\\
The \nw{intersection} $\mnw{\mathcal{K}_{SP} \cap \mathcal{K}_{PQ}}$ of $\mathcal{K}_{SP}$, $\mathcal{K}_{PQ}$ is defined over $S\uplus P\uplus Q,$ where $S,P,Q,$ are pairwise disjoint, as follows:
$\mathcal{K}_{SP} \cap \mathcal{K}_{PQ} \equivd \{ f_{SPQ} : f_{S P Q} = (f_S,h_P,g_{Q}),$
 $\textrm{ where } (f_S,h_P)\in\mathcal{K}_{SP}, (h_P,g_Q)\in\mathcal{K}_{PQ}.
\}.$\\
Thus,
$\mathcal{K}_{SP} \cap \mathcal{K}_{PQ}\equivd (\mathcal{K}_{SP} \oplus  \F_{Q}) \cap (\F_{S} \oplus \mathcal{K}_{PQ}).$\\

It is immediate from the definition of the operations that sum and intersection of
vector spaces remain vector spaces.

The following identity is useful.
%
%
%
\begin{theorem}
\label{thm:sumintersection}
Let $\V^1_A, \V^2_B, \V_S,\V'_S $ be vector spaces. Then\\
\begin{enumerate}
\item $r(\V_S)+r(\V'_S)=r(\V_S+\V'_S)+r(\V_S\cap \V'_S);$
\item $(\V^1_A+\V^2_B)^{\perp}=(\V^1_A)^{\perp}\cap (\V^2_B)^{\perp};$
\item $(\V^1_A\cap \V^2_B)^{\perp}=(\V^1_A)^{\perp}+ (\V^2_B)^{\perp}.$
\end{enumerate}
\end{theorem}

\subsection{Restriction and contraction}

The \nw{restriction}  of $\mnw{\mathcal{K}_{SP}}$ to $S$ is defined by
$\mnw{\mathcal{K}_{SP}\circ S}\equivd \{f_S:(f_S,f_P)\in \mathcal{K}_{SP}\}.$
The \nw{contraction}  of $\mnw{\mathcal{K}_{SP}}$ to $S$ is defined by
$\mnw{\mathcal{K}_{SP}\times S}\equivd \{f_S:(f_S,0_P)\in \mathcal{K}_{SP}\}.$
The sets on which we perform the contraction operation would always 
have the zero vector as a member so that the resulting set would be nonvoid.

Here again $\mnw{\mathcal{K}_{SPZ}\circ SP}$, $\mnw{\mathcal{K}_{SPZ} \times SP}$, respectively
when $S,P,Z,$ are pairwise disjoint,  denote\\  $\mnw{\mathcal{K}_{SPZ}\circ (S\uplus P)}$, $\mnw{\mathcal{K}_{SPZ} \times (S \uplus P)}.$

It is clear that restriction and contraction of vector spaces are also
vector spaces.

\subsection{Elimination of variables in linear equations}
\label{subsec:elimination}
In subsequent pages we often need to compute the constraints
on a subset of variables, given linear equations on a larger subset of 
variables. We briefly discuss the essential ideas.

Suppose we are given the linear equation
\begin{align}
\label{eqn:elimination}
\ppmatrix{C_{S} & C_{P}}\ppmatrix{x_S\\x_P}&=s,
\end{align}
and we need to compute the constraint it imposes on the variables $x_P.$
To do this we do invertible row operations on the equation to put it 
in the form
\begin{align}
\label{eqn:elimination2}
\ppmatrix{
        C_{1S} & \vdots\vdots  & C_{1P}\\
        0_{2S} & \vdots\vdots  &C_{2P}}\ppmatrix{x_S\\x_P}&=\ppmatrix{s_1\\s_2}, 
\end{align}
where rows of $C_{1S}$ are linearly independent and span the rows of $C_{S}.$
We claim that the Equation \ref{eqn:elimination} imposes the constraint
\begin{align}
\label{eqn:elimination3}
\ppmatrix{ C_{2P}}\ppmatrix{x_P}&=s_2,
\end{align}
on $x_P,$
i.e., whenever $\hat{x}_S,\hat{x}_P$ is a solution of Equation \ref{eqn:elimination}, $\hat{x}_P$  is a solution of Equation \ref{eqn:elimination3}
and whenever $\hat{x}_P$ is a solution of Equation \ref{eqn:elimination3}, there exists some $\hat{x}_S,$ such that $\hat{x}_S,\hat{x}_P$  is a solution of Equation \ref{eqn:elimination}.
The first part of the sentence is obvious. The second part follows 
because, for any given $\hat{x}_P,$
the equation 
$$ (C_{1S})x_S=-(C_{1P})\hat{x}_P+s_1$$ has a solution since rows of $C_{1S}$ are linearly independent.
It is clear that Equation \ref{eqn:elimination} has a solution iff
Equation \ref{eqn:elimination3} has a solution.
The solution always exists if the rows of $(C_{S} | C_{P})$ are linearly
independent. 
We note that if the rows of $C_{2P}$ are linearly dependent, 
Equation \ref{eqn:elimination3} has a solution iff 
$\lambda^T(C_{2P})=0$ always implies $\lambda^Ts_2=0.$

Let us now examine the special case where the right side $s$ of 
Equation \ref{eqn:elimination} is zero.

Let  $(C_S| C_P)$ be a representative matrix  for the 
vector space $\V_{SP}.$
The solution space of Equation \ref{eqn:elimination} would then be 
$\V^{\perp}_{SP}.$

Let $(C_{SP})$ denote the coefficient matrix  in Equation \ref{eqn:elimination2}.
Since $(C_{SP})$ 
is obtained from $(C_S| C_P)$ by invertible linear operations,
it is also a  representative matrix  for
vector space $\V_{SP}.$
Therefore, the rows of  $(C_{2P}),$ are linearly independent.
Further the rows of $(C_{1S})$ are independent and span the rows of $(C_S),$ by construction. Therefore, whenever $(f_S,f_P)$ is a vector in $\V_{SP},$
we must have   $f_S,$  linearly dependent on the rows of $(C_{1S}).$
Since these rows are independent, if $(0_S,f_P)$ is a vector in $\V_{SP},$
$f_P$ must be  linearly dependent on the rows of $(C_{2P}).$
By the definition of restriction and contraction of vector spaces,
we therefore conclude that $(C_{1S})$ is a representative matrix  for
$\V_{SP}\circ S$ and that  $(C_{2P})$ is a representative matrix  for
$\V_{SP}\times P.$
We say that these latter are \nw{visible} in the
representative matrix  $(C_{SP})$ of $\Vsp.$ 

From the form of the matrix $(C_{SP})$ and the discussion related to Equation \ref{eqn:elimination}
we can conclude the following:
\begin{theorem}
\label{thm:dotcrossidentity}
\begin{enumerate}
\item $r(\Vsp)=r(\Vsp\circ S)+r(\Vsp\times P);$
\item $\V_{SP}^{\perp}\circ P= (\V_{SP}\times P)^{\perp};$
\item $\V_{SP}^{\perp}\times S= (\V_{SP}\circ S)^{\perp}.$
\end{enumerate}
\end{theorem}
\begin{proof}
Part 1. is immediate from the form of the matrix $(C_{SP}).$\\ 
2. The solution space of Equation \ref{eqn:elimination}, when the right side $s=0,$ is
$\V^{\perp}_{SP}.$\\
Therefore, from the 
discussion  
related to Equation \ref{eqn:elimination} and the definition of restriction
and contraction,
we have that $\V_{SP}^{\perp}\circ P= (\V_{SP}\times P)^{\perp}.$\\
3. It is clear that the solution to Equation \ref{eqn:elimination}
has the form $(x_S,0_P)$ iff $(C_{1S})x_S=0.$
Since $(C_{1S})$ is a representative matrix of $\Vsp\circ S,$ it follows 
that $\V_{SP}^{\perp}\times S= (\V_{SP}\circ S)^{\perp}.$
\end{proof}

\subsection{Graphs}
A \nw{directed graph} $\G$ is a triplet $(V(\G),E(\G),f^d)$, where  the sets $V(\G)$, $E(\G),$ define the \nw{vertices} (or \nw{nodes}),   \nw{edges}, respectively of the graph, $f^d$ is the \nw{incidence function} which associates  with each edge an ``ordered pair'' of vertices,
called respectively its $+$ve and $-$ve endpoints.
We will refer to a {\bf directed graph in brief as a graph}, since in this paper,
we deal only with the former.
An edge is usually diagrammatically represented with an arrow going from its 
positive to its negative endpoint. The \nw{degree} of a node is the number 
of edges incident at it, with edges with single endpoints (\nw{self loops}) counted twice. An \nw{isolated} vertex has degree zero. An \nw{undirected path} between vertices $v_0,v_k$ of $\G$ is a 
sequence $(v_0,e_0,v_1, e_1, \cdots , e_{k-1},v_k),$ where 
$e_i, i=1, \cdots ,k-1, $ is incident at $v_i.$
A graph is said to be \nw{connected}, if there exists an undirected path between every pair of nodes. Otherwise it is said to be \nw{disconnected}.
A subgraph $\G_1$  of a graph $(V(\G),E(\G),f^d)$ is a graph $(V_1(\G_1),E_1(\G_1),f_1^d),$ where $V_1(\G_1)\subseteq V(\G_1),E_1(\G_1)\subseteq E(\G_1),$ and $f_1^d$ agrees with $f^d$ 
on $E_1(\G_1).$
A disconnected graph has  \nw{connected components} which are individually connected  with no edges between the components.

A \nw{loop} is a connected subgraph with the degree of each node equal to $2.$
An \nw{orientation} of a loop is a sequence of all its edges such that each
edge has a common end point with the edge succeeding it, the first edge
being treated as succeeding the last. Two orientations in which the succeeding edge to a given
edge agree are treated as the same so that there are only two possible 
orientations for a loop. 
The relative orientation of an edge 
with respect to that of the loop is positive, if the orientation 
of the loop agrees with the direction (positive node to negative node)
of the edge and negative if opposite.

A \nw{tree} subgraph of a graph is a sub-graph of the original graph with no loops. The set of edges of a tree subgraph is called a \nw{tree} and its edges are called \nw{branches}. A \nw{spanning tree} is a maximal tree with respect to the edges of a connected graph.
A \nw{cotree} of a graph is the (edge set) complement of a spanning tree of the connected graph.
A \nw{forest} of a disconnected graph is a disjoint union of the spanning trees of its connected components. The complement of a forest is called \nw{coforest}.
For simplicity, {\bf we refer to a forest (coforest) as a tree (cotree)}
 even when it is not clear that the graph is connected.

A \nw{cutset} is a minimal 
 subset of edges which when deleted from the graph increases the count of connected components by one.
Deletion of the edges of a cutset breaks exactly one of the components 
of the graph, say $\G_1,$ into two, say $\G_{11},\G_{12}.$   A cutset can be oriented in one 
of two ways corresponding to the ordered pair $(\G_{11},\G_{12})$ or the ordered pair $(\G_{12},\G_{11}).$
The relative orientation of an edge 
with respect to that of the cutset is positive if the orientation, say $(\G_{11},\G_{12}),$
of the cutset agrees with the direction (positive node to negative node)
of the edge and negative if opposite.

Let $\mathcal{G}$ be a graph with $S\equivd E(\mathcal{G})$  and let $T\subseteq S.$ Then 
$\mnw{\mathcal{G} sub (S-T)}$ denotes the graph obtained by removing the edges $T$ from $\mathcal{G}.$ 
This operation is referred to also as {\bf deletion} or
open circuiting of the edges $T.$
The graph $\mnw{\mathcal{G} \circ (S-T)}$ is obtained by removing  the isolated vertices from $\mathcal{G} sub (S-T).$
The graph $\mnw{\mathcal{G} \times (S-T)}$ is obtained by removing the edges $T$ from $\mathcal{G}$ and fusing the end vertices of the removed edges. If any isolated vertices (i.e., vertices not incident on any edges) result, they are deleted. 
Equivalently, one may first build $\mathcal{G} sub T,$
and treat each of its connected components, including the isolated nodes 
as a `supernode' of another graph with edge set $S-T.$
If any of the supernodes is isolated, it is removed.
This would result in ${\mathcal{G} \times (S-T)}.$
This operation is referred to also as {\bf contraction} or 
short circuiting of the edges $T.$
We refer to $(\G\times T)\circ W, (\G\circ T)\times W$ respectively, more simply by $\G\times T\circ W, \G\circ T\times W. $
\\
If disjoint edge sets $A,B$ 
are respectively deleted and contracted, 
the order in which these operations are performed 
can be seen to be irrelevant.
Therefore, $\G\circ (S-A)\times (S-(A\uplus B))=\G\times (S-B)\circ (S-(A\uplus B)).$
(Note that $\times, \circ $ are also used to denote vector space operations. However, the context would make clear
whether the objects involved are graphs or vector spaces.)

\nw{Kirchhoff's Voltage Law (KVL)} for a graph states that the sum of the signed voltages of
edges 
around an oriented loop is zero - the sign of the voltage of an edge 
being positive if the edge orientation agrees with the orientation of the loop
and negative if it opposes.\\
\nw{Kirchhoff's Current Law (KCL)} for a graph states that the sum of the signed currents leaving 
a node is zero, the sign of the current of an edge being positive if 
 its positive endpoint is the node in question.\\
We refer to the space of vectors $v_{S'},$ which satisfy Kirchhoff's Voltage Law (KVL) of the graph $\mathcal{G},$
by $\mnw{\V^v(\mathcal{G})}$ and to the space of vectors $i_{S"},$ which satisfy Kirchhoff's Current Law (KCL) of the graph $\mathcal{G},$
by $\mnw{\V^i(\mathcal{G})}.$
{\bf These vector spaces will, unless otherwise stated, be taken as  
over $\Re.$}

The following are useful results on vector spaces associated with graphs.
\begin{theorem}
\label{thm:tellegen}
{\bf Tellegen's Theorem} (\cite{tellegen},\cite{penfield})  $\V^i(\mathcal{G})= \V^v(\mathcal{G})^{\perp}.$
\end{theorem}
\begin{lemma}
\label{lem:minorgraphvectorspace}
\cite{tutte} Let $\G$ be a graph on edge set $S.$ 
Let $W\subseteq T\subseteq S.$
\begin{enumerate}
\item $ \V^v(\mathcal{G}\circ T)= (\V^v(\mathcal{G}))\circ T, \ \ \  \V^v(\mathcal{G}\times T)= (\V^v(\mathcal{G}))\times T,\ \ \V^v(\mathcal{G}\circ T\times W)= (\V^v(\mathcal{G}))\circ T \times W;$
\item $ \V^i(\mathcal{G}\circ T)= (\V^i(\mathcal{G}))\times T, \ \ \  \V^i(\mathcal{G}\times T)= (\V^i(\mathcal{G}))\circ T,
\V^i(\mathcal{G}\times T\circ W)= (\V^i(\mathcal{G}))\circ T \times W.$
\end{enumerate}
\end{lemma}
(Proof also available at \cite{HNarayanan2009}).\\
(Note that $\times,\circ$ are graph operations on the left side of the equations
and vector space operations on the right side.)
\subsection{Networks and multiports}
\label{subsec:networks}
A (static) {\bf electrical network $\mathcal{N},$} or a `network' in short, is a pair $(\mathcal{G},\mathcal{K}),$ where $\mathcal{G}\equivd (V(\G),E(\G),f^d)$ is a directed graph
and $\mathcal{K}, $ called the \nw{device characteristic} of the network, is a collection of pairs of vectors  $(v_{S'},i_{S"}),S\equivd E(\G)$ where  $v_{S'},i_{S"}$ are real or complex vectors on the edge set of the graph. 
In this paper,  we deal only with affine device characteristics 
and with real vectors, unless otherwise stated. When the device characteristic $\K_{S'S"}$ is affine,
we say the network is \nw{linear}. If $\V_{S'S"}$ is the vector space 
translate of $\K_{S'S"},$ we say that $\K_{S'S"}$ is the \nw{source accompanied}
 form of $\V_{S'S"}.$ An affine space $\A_{S'S"}$ is said to be \nw{proper} 
iff its vector space
translate $\V_{S'S"}$ has dimension $|S'|=|S"|.$ 

Let $S$ denote the set of edges of the graph of the network and let
$\{S_1, \cdots , S_k\} $ be a partition of $S,$ each block $S_j$ being an 
\nw{individual device}. Let $S',S"$ be copies of  $S,$
with $e',e"$ corresponding to edge $e.$ The device characteristic would usually have the form $\bigoplus \K_{S_j'S_j"},$  
defined by $(B_j'v_{S_j'}+Q_j"i_{S_j"})=s_j,$
with rows of $(B_j'|Q_j")$ being linearly independent.
\\
We say $S_j$ is  a set of \nw{norators}
iff $\K_{S_j'S_j"}\equivd \F_{S_j'S_j"},$ i.e., there are no constraints on
$v_{S_j'},i_{S_j"}.$\\
A {\bf solution} of $\mathcal{N}\equivd (\G,\K)$ on graph $\G\equivd (V(\G),E(\G),f^d)$ is a pair
 $(v_{S'},i_{S"}),S\equivd E(\G)$ satisfying\\
$v_{S'}\in \V^v(\mathcal{G}),\ \ i_{S"} \in \V^i(\mathcal{G})$
  (KVL,KCL)  and $(v_{S'},i_{S"})\in \mathcal{K}.$
The KVL,KCL conditions are also called \nw{topological} constraints.
Let  let $S',S"$ be disjoint copies of $S,$
let $\V_{S'}\equivd \V^v(\mathcal{G}),(\V^{\perp}_{S'})_{S"}= \V^i(\mathcal{G}),
\ \K_{S'S"}$ be the device characteristic of $\N.$
The set of solutions of  $\mathcal{N}$ may be written, using the extended
definition of intersection as
$$\V_{S'}\cap (\V^{\perp}_{S'})_{S"}\cap \K_{S'S"}=[\V_{S'}\oplus (\V^{\perp}_{S'})_{S"}]\cap \K_{S'S"}.$$
This has the form `[Solution set of topological constraints] $\cap$ [Device characteristics]'.\\
The \nw{power absorbed} at an edge $e\in E$ corresponding to a solution 
$(v_{E'},i_{E"})$ is given by $v_{E'}(e')\times i_{E"}(e").$

A \nw{multiport} $\mathcal{N}_P$ is a network with some subset $P$ of its
edges which are norators, specified as `ports'.
The multiport is said to be \nw{linear} iff its device characteristic is affine.
Let $\N_P$ be on graph $\G_{SP}$ with device characteristic $\K.$
Let $\V_{S'P'}\equivd (\V^v(\G_{SP}))_{S'P'}, (\V^{\perp}_{S'P'})_{S"P"}= (\V^i(\G_{SP}))_{S"P"},$ and let $\K_{S'S"}, $ be the affine device characteristic 
on the edge set $S.$ The device characteristic of $\mathcal{N}_P$
would be $\K\equivd \K_{S'S"} \oplus \F_{P'P"}.$
For simplicity we would refer to $\K_{S'S"}$ as the device characteristic 
of $\N_P.$
\\
The set of solutions of  $\mathcal{N}_P$ may be writen, using the extended
definition of intersection as
$$\V_{S'P'}\cap (\V^{\perp}_{S'P'})_{S"P"}\cap \K_{S'S"}=[\V_{S'P'}\oplus (\V^{\perp}_{S'P'})_{S"P"}]\cap \K_{S'S"}.$$
We say the multiport is \nw{consistent} iff its set of solutions 
is nonvoid.


The multiport $\mathcal{N}_P$ would impose a relationship 
$([\V_{S'P'}\oplus (\V^{\perp}_{S'P'})_{S"P"}]\cap \K_{S'S"})\circ P'P",$ between
$v_{P'},i_{P"}.$ The \nw{multiport behaviour}  (port behaviour for short) $\breve{\K}_{P'P"}$ at $P,$  of $\N_P,$ is defined 
by\\ 
$\breve{\K}_{P'P"}\equivd [([\V_{S'P'}\oplus (\V^{\perp}_{S'P'})_{S"P"}]\cap \K_{S'S"})\circ P'P")]_{P'(-P")}= ([\V_{S'P'}\oplus (\V^{\perp}_{S'P'})_{S"(-P")}]\cap \K_{S'S"})\circ P'P".$ 
\\
When the device characteristic of $\mathcal{N}_P$ is affine, its {port behaviour} $\breve{\K}_{P'P"}$ at $P$ would be 
affine if it were not void.
%
%
%
\begin{remark}
Note that,
if the multiport
is a single port edge in parallel with a positive resistor $R,$\\
$([\V_{S'P'}\oplus (\V^{\perp}_{S'P'})_{S"(-P")}]\cap \K_{S'S"})\circ P'P"$
would be the solution of $v_{P'}= -Ri_{P"}.$
But then $\breve{\K}_{P'P"},$ as defined, would be the solution of $v_{P'}= Ri_{P"}.$
\\
This change is required if we wish to use the port behaviour as the device
characteristic of another multiport  which is on  some graph $\G_{PQ}$
with $P$ as the set of internal edges and $Q$ as the set of ports.
\\
Observe that  a multiport made up of passive (power absorbed always nonnegative) devices will  have passive
port characteristics which would not be possible without the change of sign.
\end{remark}



\section{Matched and Skewed Composition}
\label{sec:matched}
In this section we introduce an operation between collections of vectors
motivated by the connection of multiport behaviours across ports.

Let $\Ksp,\Kpq,$ be collections of vectors respectively on $S\uplus P,P\uplus Q,$ with $S,P,Q,$ being pairwise disjoint.

The \nw{matched composition} $\mnw{\mathcal{K}_{SP} \leftrightarrow \mathcal{K}_{PQ}}$ is on $S\uplus Q$ and is defined as follows:
\begin{align*}
 \mathcal{K}_{SP} \leftrightarrow \mathcal{K}_{PQ} 
  &\equivd \{
                 (f_S,g_Q): (f_S,h_P)\in \Ksp, 
(h_P,g_Q)\in \Kpq\}.
\end{align*}
Matched composition is referred to as matched sum in \cite{HNarayanan1997}.


The \nw{skewed composition} $\mnw{\mathcal{K}_{SP} \leftrightarrow \mathcal{K}_{PQ}}$ is on $S\uplus Q$ and is defined as follows:
\begin{align*}
 \mathcal{K}_{SP} \rightleftharpoons \mathcal{K}_{PQ} 
  &\equivd \{
                 (f_S,g_Q): (f_S,h_P)\in \Ksp, 
(-h_P,g_Q)\in \Kpq\}.
\ \mbox{Note that}
\end{align*}
$$\mathcal{K}_{SP} \rightleftharpoons \mathcal{K}_{PQ}\ \ \ =\ \ \ \mathcal{K}_{SP} \lrar \mathcal{K}_{(-P)Q}.$$
When $S$, $Y$ are disjoint, both the matched and skewed composition of
$\K_S,\K_Y,$ correspond to the direct sum $\K_S\oplus \K_Y$.
It is clear from the definition of matched composition and that of restriction
and contraction, that
$ \Ks\circ (S-T) =\Ks\lrar \F_T, \Ks\times (S-T)=\Ks\lrar \0_T, T\subseteq S.$
When $\mathcal{K}_{SP}$, $\mathcal{K}_P$ are vector spaces, observe that $(\mathcal{K}_{SP}\leftrightarrow \mathcal{K}_P) = (\mathcal{K}_{SP}\rightleftharpoons \mathcal{K}_P).$
When $S,P,Z,$ are pairwise disjoint,
we have\\
$(\mathcal{K}_{SPZ} \leftrightarrow  \mathcal{K}_{S}) \leftrightarrow  \mathcal{K}_{P} = (\mathcal{K}_{SPZ} \leftrightarrow  \mathcal{K}_{P}) \leftrightarrow  \mathcal{K}_{S} = \mathcal{K}_{SPZ} \leftrightarrow  (\mathcal{K}_{S} \oplus  \mathcal{K}_{P}).$
When $ \mathcal{K}_{S}\equivd \0_S, \mathcal{K}_{P}\equivd \mathcal{K}_{SPZ}\circ P,$ the above reduces to
$ \mathcal{K}_{SPZ}\times PZ\circ Z\equaln \mathcal{K}_{SPZ}\circ SZ\times Z.$
Such an object is called a \nw{minor} of $\mathcal{K}_{SPZ}.$
In the special case where $Y\subseteq S$, the matched composition $\Ks\lrar \K_Y,$ is called the
\nw{generalized minor} of $\mathcal{K}_S $
with respect to $\mathcal{K}_Y$.

The following result is immediate from the definition of matched
and skewed composition.
\begin{theorem}
\label{thm:matchedprop}
Let $\Ksp,\Kpq$ be collections of vectors on $S\uplus P,  P\uplus Q,$ respectively.
Then,
$$\Ksp\lrar \Kpq= (\Ksp\cap \Kpq)\circ SQ;\ \ 
\Ksp\rightleftharpoons \Kpq\ \ =\ \  (\Ksp\cap (\Kpq)_{(-P)Q})\circ SQ;$$
$$\Ksp\lrar \Kpq = (\Ksp+(\Kpq)_{(-P)Q})\times SQ;\ \ 
\Ksp\rightleftharpoons \Kpq\ \ =\ \  (\Ksp+ \Kpq)\times SQ.$$
\end{theorem}

\subsection{Multiport connection and matched and skewed composition}
\label{subsec:connectioncomposition}
In this subsection we illustrate the notions of matched and skewed 
composition through the operation of connecting multiports.

From the definition of port behaviour of a multiport $\N_P$
 and that of the `$\lrar $' operation, we have the following lemma.
\begin{lemma}
\label{lem:behaviourlrar}
Let the set of solutions of the multiport $\mathcal{N}_P$  restricted to
$P'\uplus P"$ be\\
$(\V_{S'P'}\cap (\V^{\perp}_{S'P'})_{S"P"}\cap \K_{S'S"})\circ P'P".$
 Then the  {port behaviour} $\breve{\K}_{P'P"}$ at $P,$  of $\N_P,$ is given
by\\
$\breve{\K}_{P'P"}\equivd ((\V_{S'P'}\cap (\V^{\perp}_{S'P'})_{S"P"}\cap \K_{S'S"})\circ P'P")_{P'(-P")}$\\$= (\V_{S'P'}\cap (\V^{\perp}_{S'P'})_{S"(-P")}\cap \K_{S'S"})\circ P'P"$
$=
(\V_{S'P'}\oplus (\V^{\perp}_{S'P'})_{S"(-P")})\lrar  \K_{S'S"}.$
\end{lemma}
\begin{figure}
\begin{center}
 \includegraphics[width=3.65in]{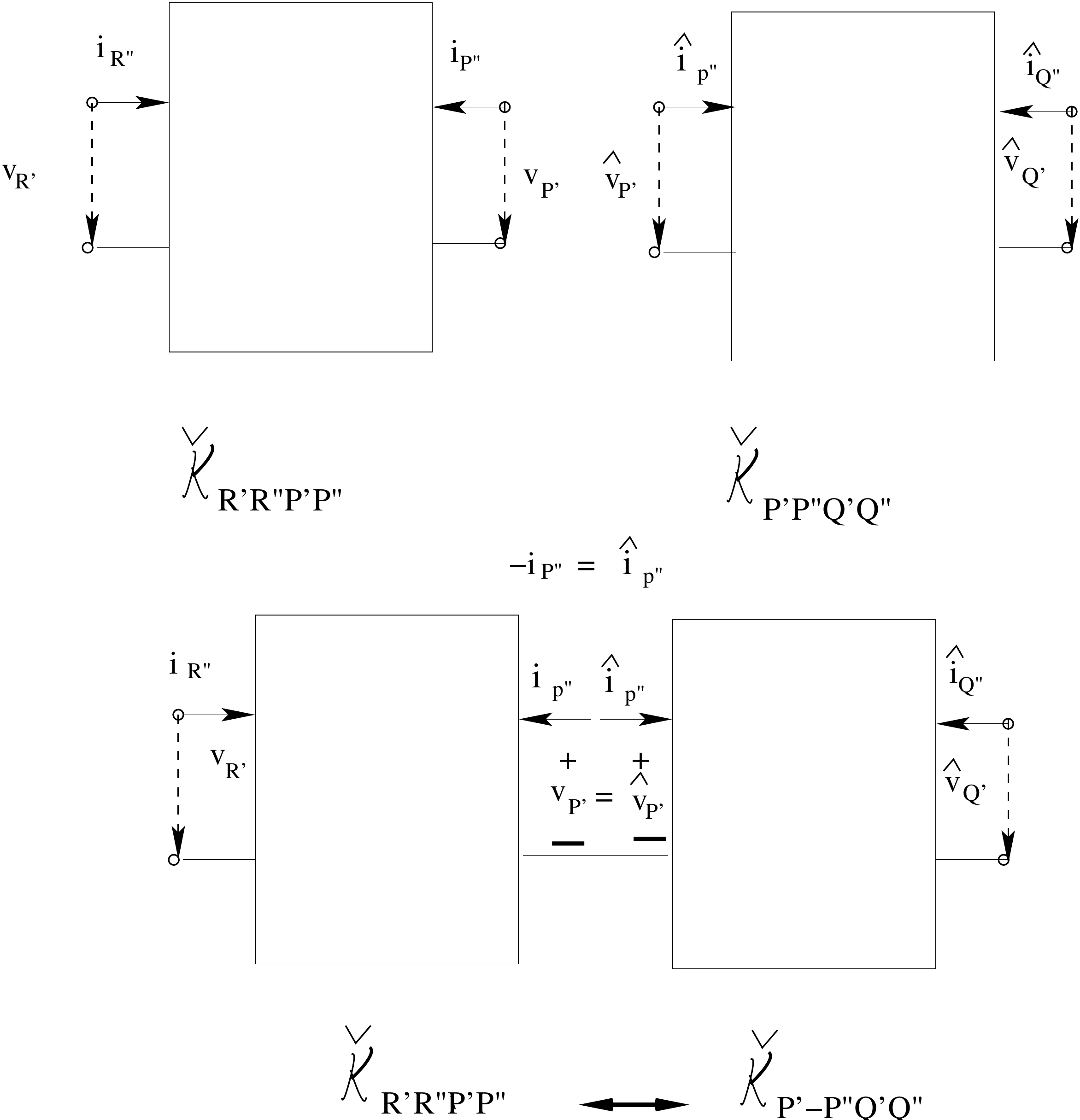}
 \caption{Connection of multiport behaviours and matched ($P',P'$) and skewed ($P",-P"$) composition
}
\label{fig:multiportfinal}
\end{center}
\end{figure}

\begin{example}
\label{eg:connection}
1. Consider the multiport behaviours  $\breve{\K}_{R'R"P'P"}, \breve{\K}_{P'P"Q'Q"}$ in Figure \ref{fig:multiportfinal}.
Connecting them together at the ports $P$  amounts to
making the voltages at $P,$ i.e., $v_{P'},\hat{v}_{P'}$  same in both the multiports\\ $\breve{\K}_{R'R"P'P"}, \breve{\K}_{P'P"Q'Q"}$ and the currents $i_{P"}$ of $\breve{\K}_{R'R"P'P"}$ equal to negative of 
the current $\hat{i}_{P"}$ of $\breve{\K}_{P'P"Q'Q"}.$
The multiport behaviour that results can be denoted by
$\breve{\K}_{R'R"Q'Q"}\equivd \breve{\K}_{R'R"P'P"} \lrar \breve{\K}_{P'(-P")Q'Q"}.$
\end{example}
We say  multiport behaviours $\K_{R'R"P'P"}, \K_{P'P"Q'Q"}$ are \nw{connected across} ports 
$P$ to \nw{yield} the multiport  behaviour $\K_{P'P"Q'Q"}$
iff $\K_{R'R"Q'Q"}= \K_{R'R"P'P"} \lrar \K_{P'(-P")Q'Q"}.$\\
The following lemma, which is immediate from the definition of connection 
of behaviours, gives an equivalent way of looking at  the operation.
\begin{lemma}
\label{lem:connectionidealtransformer}
Let $\tilde{P}$ be a copy of $P,$
with $P\equivd\{e_1, \cdots , e_k\}, \tilde{P}\equivd \{\tilde{e}_1, \cdots , \tilde{e}_k\}, e_i,\tilde{e}_i$ being copies of each other.\\
Further, let $R',R",P',P",\tilde{P}',\tilde{P}",Q',Q",$ be pairwise disjoint.\\
Let $ \K_{\tilde{P}'\tilde{P}"Q'Q"}\equivd (\K_{P'P"Q'Q"})_{\tilde{P}'\tilde{P}"Q'Q"}.$\\
Let $\T^{P\tilde{P}}$ denote the solution space of the equations
$v_{\tilde{e}_i'}= v_{e_i'}; i_{\tilde{e}_i"}= -i_{e_i"}, i=1, \cdots, k.$\\
Then $\K_{R'R"Q'Q"}\equivd \K_{R'R"P'P"} \lrar \K_{P'(-P")Q'Q"}=[\K_{R'R"P'P"} \oplus  \K_{\tilde{P}'\tilde{P}"Q'Q"}]\lrar \T^{P\tilde{P}}.$
\end{lemma}
In line with the idea of connection through $\T^{P\tilde{P}}$ in the above lemma, we can define connection 
through an arbitrary device $\K^{P\tilde{P}}_{P'\tilde{P}'P"\tilde{P}"}.$\\
Let the multiports $\N_{RP},{\N}_{\tilde{P}Q}$ be on graphs $\G_{RSP},\G_{\tilde{P}MQ}$ respectively, with the primed and double primed sets obtained from 
$R,S,P,\tilde{P},M,Q,$ being pairwise disjoint,
and let them have device characteristics $\K^{S},{\K}^{M}$ respectively.
Let $\K^{P\tilde{P}}$ denote a collection of vectors $\K^{P\tilde{P}}_{P'\tilde{P}'P"\tilde{P}"}.$\\
The multiport $\mnw{[\N_{RP}\oplus {\N}_{\tilde{P}Q}]\cap \K^{P\tilde{P}}},$ 
with ports $R\uplus Q$ obtained by \nw{connecting $\N_{RP},{\N}_{\tilde{P}Q}$
through $\K^{P\tilde{P}}$},
is on graph $\G_{RSP}\oplus\G_{\tilde{P}MQ}$ with device characteristic
$\K^{S}\oplus {\K}^{M}\oplus  \K^{P\tilde{P}}$ (see Figure \ref{fig:networkmultiportconnection}).\\
\begin{figure}
\begin{center}
 \includegraphics[width=4in]{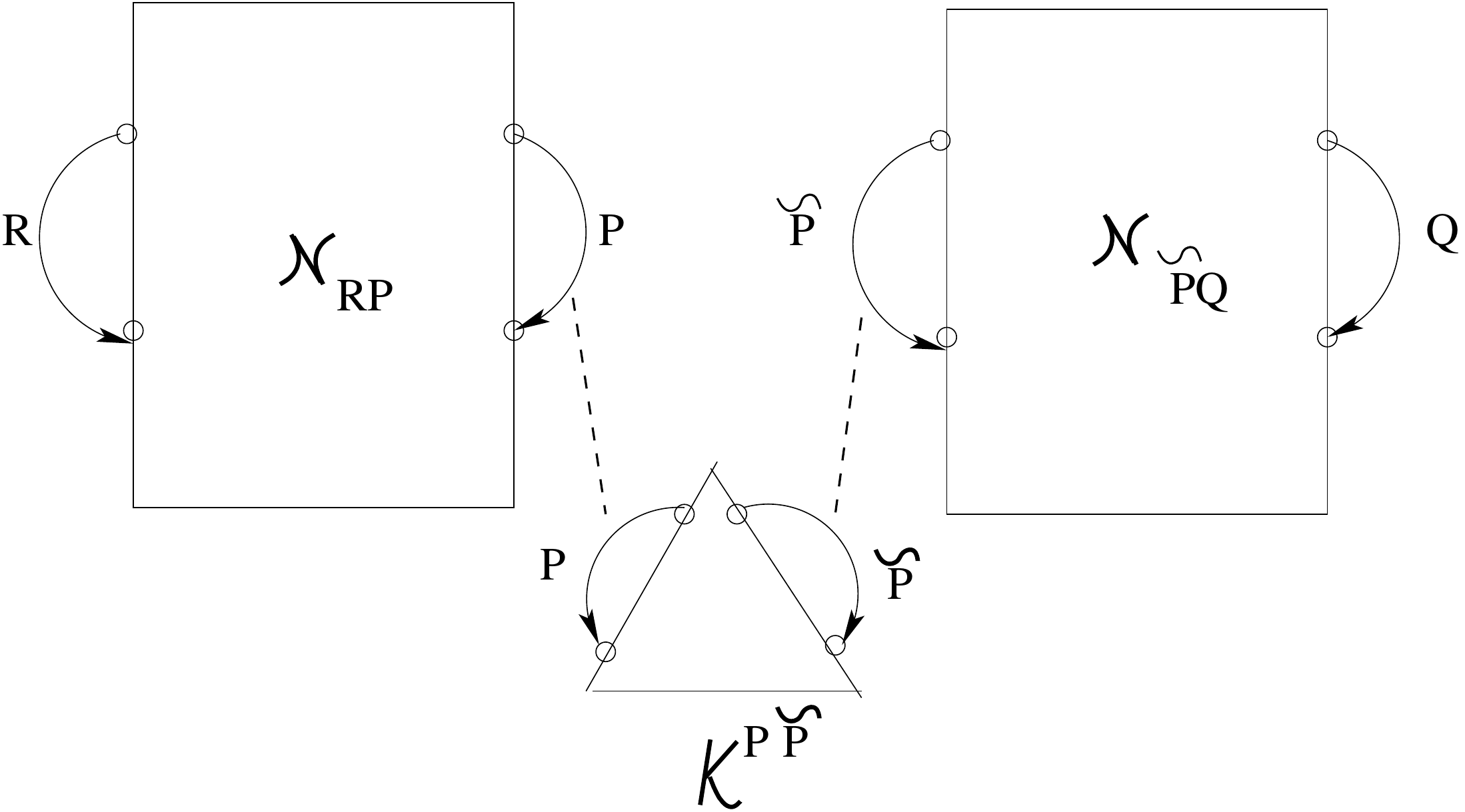}
 \caption{Multiport network $\mnw{\N_{RQ}\equivd [\N_{RP}\oplus {\N}_{\tilde{P}Q}]\cap \K^{P\tilde{P}}}$
}
\label{fig:networkmultiportconnection}
\end{center}
\end{figure}
When $R,Q$ are void, $[\N_{RP}\oplus {\N}_{\tilde{P}Q}]\cap \K^{P\tilde{P}}$ would reduce to $[\N_{P}\oplus {\N}_{\tilde{P}}]\cap \K^{P\tilde{P}}$  and would be a network without ports.
In this case we say \nw{the multiport $\N_{P}$ is terminated by ${\N}_{\tilde{P}}$ through $\K^{P\tilde{P}}.$}\\
This network is on  graph $\G_{SP}\oplus \G_{\tilde{P}M}$ with device
characteristic $\K^S_{S'S"}
\oplus  {\K}^M_{M'M"}\oplus \K^{P\tilde{P}}_{P'P"\tilde{P}'\tilde{P}"}.$

For ready reference, we summarize facts about the multiport ${[\N_{RP}\oplus {\N}_{\tilde{P}Q}]\cap \K^{P\tilde{P}}},$ and the network $[\N_{P}\oplus {\N}_{\tilde{P}}]\cap \K^{P\tilde{P}}$ in the theorem below. The result is essentially 
a restatement of the definition of connection of multiports.
%
%
%
%
%
\begin{theorem}
\begin{enumerate}
\item The set of solutions of $\N_{RQ}\equivd [\N_{RP}\oplus {\N}_{\tilde{P}Q}]\cap \K^{P\tilde{P}}$ is \\$[[\V_{S'R'P'}\oplus (\V^{\perp}_{S'R'P'})_{S"R"P"}]\cap \K^S_{S'S"}]
\oplus [[\V_{M'\tilde{P}'Q'}\oplus (\V^{\perp}_{M'\tilde{P}'Q'})_{M"\tilde{P}"Q"}]\cap \tilde{\K}^M_{M'M"}]\cap \K^{P\tilde{P}}_{P'P"\tilde{P}'\tilde{P}"}$
\\$= [\V_{S'R'P'}\oplus (\V^{\perp}_{S'R'P'})_{S"R"P"}\oplus \V_{M'\tilde{P}'Q'}\oplus (\V^{\perp}_{M'\tilde{P}'Q'})_{M"\tilde{P}"Q"}]\cap [\K^S_{S'S"}\oplus
\tilde{\K}^M_{M'M"}\oplus \K^{P\tilde{P}}_{P'P"\tilde{P}'\tilde{P}"}],$
where $\V_{S'R'P'}\equivd (\V^v(\G_{SRP}))_{S'R'P'}, \V_{M'\tilde{P}'Q'}\equivd
(\V^v(\G_{M\tilde{P}Q}))_{M'\tilde{P}'Q'}.$
\item Let $\breve{\K}_{R'P'R"P"}, \breve{\K}_{\tilde{P}'Q'\tilde{P}"Q"}$ 
be the port behaviours of $\N_{RP},{\N}_{\tilde{P}Q}$ respectively. Then the port behaviour 
of $\N_{RQ}$ is given by $\breve{\K}_{R'Q'R"Q"}=[\breve{\K}_{R'P'R"P"}\oplus \breve{\K}_{\tilde{P}'Q'\tilde{P}"Q"}]\lrar (\K^{P\tilde{P}}_{P'P"\tilde{P}'\tilde{P}"})_{P'(-P")\tilde{P}'(-\tilde{P}")}.$
\item If $\N_{RP},{\N}_{\tilde{P}Q}$ are linear multiports, so would 
the multiport $\N_{RQ}$
be, with affine set of solutions and affine port behaviour.
\item If $\N_P,{\N}_{\tilde{P}}$ are multiports 
then  the network $[\N_{P}\oplus {\N}_{\tilde{P}}]\cap \K^{P\tilde{P}}$ has, as the set of solutions,\\ 
$[[\V_{S'P'}\oplus (\V^{\perp}_{S'P'})_{S"P"}]\cap \K^S_{S'S"}]
\oplus [\V_{M'\tilde{P}'}\oplus (\V^{\perp}_{M'\tilde{P}'Q'})_{M"\tilde{P}"}]\cap \tilde{\K}^M_{M'M"}]\cap \K^{P\tilde{P}}_{P'P"\tilde{P}'\tilde{P}"}$\\
$=[\V_{S'P'}\oplus (\V^{\perp}_{S'P'})_{S"P"}\oplus \V_{M'\tilde{P}'}\oplus (\V^{\perp}_{M'\tilde{P}'})_{M"\tilde{P}"}] \cap [\K^S_{S'S"}
\oplus  {\K}^M_{M'M"}\oplus \K^{P\tilde{P}}_{P'P"\tilde{P}'\tilde{P}"}],$\\
where $\V_{S'P'}\equivd (\V^v(\G_{SP}))_{S'P'}, \V_{\tilde{P}'Q'}\equivd
(\V^v(\G_{\tilde{P}Q}))_{\tilde{P}'Q'}.$

\end{enumerate}
\end{theorem}
\begin{proof}
We prove only part 2.\\
$\breve{\K}_{R'Q'R"Q"}=[[(\breve{\K}_{R'P'R"P"})_{R'P'(-R")(-P")}\oplus (\breve{\K}_{\tilde{P}'Q'\tilde{P}"Q"})_{\tilde{P}'Q'(-\tilde{P}")(-Q")}]\lrar (\K^{P\tilde{P}}_{P'P"\tilde{P}'\tilde{P}"})]_{R'(-R")Q'(-Q")}$
\\$=[(\breve{\K}_{R'P'R"P"})_{R'P'R"(-P")}\oplus (\breve{\K}_{\tilde{P}'Q'\tilde{P}"Q"})_{\tilde{P}'Q'(-\tilde{P}")Q"}]\lrar \K^{P\tilde{P}}_{P'P"\tilde{P}'\tilde{P}"}$
\\$=[\breve{\K}_{R'P'R"P"}\oplus \breve{\K}_{\tilde{P}'Q'\tilde{P}"Q"})]\lrar (\K^{P\tilde{P}}_{P'P"\tilde{P}'\tilde{P}"})_{P'(-P")\tilde{P}(-\tilde{P}")}.$
\end{proof}

\section{Implicit Inversion Theorem and its application to multiports}
\label{sec:iit}
In this section we state one of the two basic results of implicit linear
algebra and present some applications for multiports. \ref{sec:proofsiitidt} gives generalizations and their proofs.

\begin{theorem}
\label{thm:IITlinear}
Consider the equation
$$\Vsp \lrar \Vpq =\Vsq, $$
where $\Vsp, \Vpq, \Vsq$ are vector spaces respectively on $S\uplus P,P\uplus Q,S\uplus Q,$ with $S,P,Q,$ being pairwise disjoint.
We then have the following.
\begin{enumerate}
\item given $\Vsp, \Vsq,$  there exists $ \Vpq, $ satisfying the equation only if  $\Vsp\circ S\supseteq \Vsq \circ S$ and $\Vsp\times S \subseteq \Vsq \times S.$
\item given $\Vsp, \Vsq,$ if  $\Vsp\circ S\supseteq \Vsq \circ S$ and $\Vsp\times S \subseteq \Vsq \times S, $ then $\hat{\V}_{PQ}\equivd \Vsp \lrar \Vsq $ 
satisfies the equation.
\item given $\Vsp, \Vsq,$ assuming that the equation  $\Vsp \lrar \Vpq =\Vsq $
is satisfied by some $\hat{\V}_{PQ}, $ it is unique if  the additional conditions 
 $\Vsp\circ P\supseteq \Vpq \circ P$ and $\Vsp\times P\subseteq \Vpq \times P$
are imposed.
\end{enumerate}
\end{theorem}

The following is  a useful affine space version of Theorem \ref{thm:IITlinear}.
\begin{theorem}
\label{thm:IIT2}
Let $\Asp,\Apq$ be affine spaces on $S\uplus P,P\uplus Q,$ where $S,P,Q$ 
are pairwise disjoint sets. Let $\Vsp,\Vpq$ respectively, be the vector 
space translates of $\Asp,\Apq.$ Let $\Asp\lrar \Apq$ be \nw{nonvoid} and let\\
$\alpha_{SQ}\in \Asp\lrar \Apq.$ Then,\\
1. $\Asp\lrar \Apq = \alpha_{SQ}+(\Vsp\lrar \Vpq).$\\ 
2. $\Apq = \Asp\lrar (\Asp\lrar \Apq)$ iff
$\Vsp\circ P\supseteq \Vpq \circ P$ and $\Vsp\times P\subseteq \Vpq \times P.$
\end{theorem}

\begin{proof}
1. Since $\alpha_{SQ}\in \Asp\lrar \Apq,$ there exist 
$\alpha _{SP}\in \Asp, \alpha _{PQ}\in \Apq,$ such that
$\alpha _{SP}|_P= \alpha _{PQ}|_P, \alpha _{SQ}=(\alpha _{SP}|_S,\alpha _{PQ}|_Q).$
 
Let $x_{SQ}\in \Vsp\lrar \Vpq.$
Then there exist  $x_{SP}\in \Vsp, x_{PQ}\in \Vpq,$ such that 
$x_{SP}|_P= x_{PQ}|_P, x_{SQ}=(x_{SP}|_S,x_{PQ}|_Q).$
Now $(\alpha _{SP}+x_{SP})\in \Asp, (\alpha _{PQ}+x_{PQ})\in \Apq, (\alpha _{SP}+x_{SP})|_P=(\alpha _{PQ}+x_{PQ})|_P,$
so that
$(\alpha _{SQ}+x_{SQ})\in \Asq.$
We conclude that $\alpha_{SQ}+(\Vsp\lrar \Vpq)\subseteq (\Asp\lrar \Apq).$

Next let $\beta_{SQ}\in \Asp\lrar \Apq.$ There exist
$\beta _{SP}\in \Asp, \beta _{PQ}\in \Apq,$ such that
$\beta _{SP}|_P= \beta _{PQ}|_P, \beta _{SQ}=(\beta _{SP}|_S,\beta _{PQ}|_Q).$
It follows that $(\alpha_{SP}-\beta_{SP})\in \Vsp, (\alpha_{PQ}-\beta_{PQ})\in \Vpq ,(\alpha_{SP}-\beta_{SP})|_P=(\alpha_{PQ}-\beta_{PQ})|_P $ and $(\alpha_{SQ}-\beta_{SQ})\in (\Vsp\lrar \Vpq).$

Thus $\alpha _{SQ}-(\alpha_{SQ}-\beta_{SQ}) = \beta_{SQ}\in \alpha_{SQ}+(\Vsp\lrar \Vpq).$

We conclude that $\Asp\lrar \Apq= \alpha_{SQ}+(\Vsp\lrar \Vpq).$

2. We are given that $\Asp\lrar \Apq$ is nonvoid so that $\Asp,\Apq$
are nonvoid.
In this case we have by part 1 above, that the vector space translate
of $\Asp\lrar \Apq$ is $\Vsp\lrar \Vpq,$
where $\Vsp,\Vpq,$ are respectively the vector space translates of $\Asp,\Apq.$
Thus, again by part 1,  $\Apq = \Asp\lrar (\Asp\lrar \Apq)$ iff\\
$\Vpq = \Vsp\lrar (\Vsp\lrar \Vpq).$
The result now follows from part 2 of Theorem \ref{thm:IITlinear}.
\end{proof}
\begin{remark}
 Testing whether $\Asp\lrar \Apq$ is nonvoid is equivalent 
to checking whether a set of linear equations has a solution.
Let $\Vsp,\Vpq$ be the vector space translates of $\Asp, \Apq,$
and let $(B_S|B_P),(A_P|A_Q)$ be representative matrices of $\Vsp,\Vpq.$
Let $\Asp= \alpha_{SP}+\Vsp, \Apq= \beta_{PQ}+\Vpq.$\\
Let $\alpha_{SP}^T=(\alpha_{S}^T,\alpha_{P}^T), \beta_{PQ}^T=(\beta_{P}^T,\beta_{Q}^T).$
Consider the equation
\begin{align}
\label{eqn:affinenonvoid}
\lambda^T(B_P)+\alpha_P^T=\sigma^T(A_P)+\beta_{P}^T.
\end{align}
Let $\hat{\lambda}, \hat{\sigma}$ be a solution to this equation.
Then the vector $(\hat{\lambda}^T(B_S|B_P)+(\alpha_S^T, \alpha_P^T))\in \Asp, $\\$(\hat{\sigma}^T(A_P|A_Q)+(\beta_P^T,\beta_Q^T))\in \Apq, (\hat{\lambda}^T(B_S)+\alpha_S^T, \hat{\sigma}^T(A_Q)+\beta_Q^T)\in \Asp\lrar \Apq.$\\
Conversely, if  $\Asp\lrar \Apq$ is nonvoid, Equation \ref{eqn:affinenonvoid}
 has a solution.
\end{remark}
\subsection{Affine multiport connection}
We have an immediate corollary of Theorem \ref{thm:IIT2} for multiports, whose device characteristics 
are affine spaces.

\begin{corollary}
\label{cor:IIT2}
Let affine multiport behaviours $\breve{\A}_{R'R"P'P"},\breve{\A}_{P'P"Q'Q"}$ be connected across ports $P$
to yield the nonvoid multiport behaviour $\breve{\A}_{R'R"Q'Q"}.$
Let $\breve{\A}_{R'R"P'P"},\breve{\A}_{R'R"Q'Q"}$ be known and let
$\breve{\V}_{R'R"P'P"},\breve{\V}_{P'P"Q'Q"},$
be respectively the vector space translates of $\breve{\A}_{R'R"P'P"},\breve{\A}_{P'P"Q'Q"}.$
\\
Then 
\begin{enumerate}
\item the vector space
translate of $\breve{\A}_{R'R"Q'Q"}$ is $\breve{\V}_{R'R"P'P"}\lrar (\breve{\V}_{P'P"Q'Q"})_{P'(-P")Q'Q"};$ 
\item $\breve{\A}_{P'P"Q'Q"}$ can
be uniquely determined, being equal to $[\breve{\A}_{R'R"P'P"}\lrar \breve{\A}_{R'R"Q'Q"}]_{P'(-P")Q'Q"},$\\ iff
 the additional conditions
$(\breve{\V}_{R'R"P'P"}\circ P'P")_{P'(-P")}\supseteq \breve{\V}_{P'P"Q'Q"}\circ P'P"$\\ and
$(\breve{\V}_{R'R"P'P"}\times P'P")_{P'(-P")}\subseteq \breve{\V}_{P'P"Q'Q"}\times P'P",$
are imposed.
\end{enumerate}

\end{corollary}


We note that 
(see part 1 of Example \ref{eg:connection}),
if multiports $\mathcal{N}_P,$ $\mathcal{N}'_P$ have the same graph, but 
device characteristics $\A_{S'S"},\V_{S'S"}$ respectively and further $\V_{S'S"}$ is the vector space translate of $\A_{S'S"},$ then the {port behaviour} $\breve{\A}_{P'P"}$
of multiport $\mathcal{N}_P$ is
 $
(\V_{S'P'}\oplus (\V^{\perp}_{S'P'})_{S"(-P")})\lrar  \A_{S'S"}$
and  that of multiport $\mathcal{N}'_P$ 
is 
 $
(\V_{S'P'}\oplus (\V^{\perp}_{S'P'})_{S"(-P")})\lrar  \V_{S'S"}.$
The following results are  immediate consequences of Theorem \ref{thm:IIT2}.
\begin{theorem}
\label{thm:translatemultiport}
 Let $\N^1_P,{\N}^2_P$ be multiports on the same graph $\G_{SP}$
but with device characteristics \\
$\A_{S'S"}, {\V}_{S'S"},$ respectively where ${\V}_{S'S"},$
is the vector space translate of the affine space $\A_{S'S"},$
and with port behaviours $\breve{\A}_{P'P"}, \breve{\V}_{P'P"},$ respectively.
\\
If $\breve{\A}_{P'P"}\ne \emptyset ,$ then $\breve{\V}_{P'P"}=
([(\V^v({\G_{SP}}))_{S'P'}\oplus (\V^i({\G_{SP}}))_{S"P"}]\lrar {\V}_{S'S"})_{P'(-P")},$
is the vector space translate of $\breve{\A}_{P'P"}.$
\end{theorem}
Other applications of Theorem \ref{thm:IITlinear} may be found in 
\cite{HNarayanan1997,HNPS2013, narayanan2016}.

\section{Implicit Duality Theorem and its application to multiports}
\label{sec:idt}
In this section we state the second basic result of ILA and present some 
applications.
Implicit Duality Theorem is a part of network theory folklore.
However, its applications are insufficiently emphasized in the literature.
\begin{theorem}
\label{thm:idt0}
{\bf Implicit Duality Theorem} 
Let $\Vsp, \Vpq$ be vector spaces respectively on $S\uplus P,P\uplus Q,$ with $S,P,Q,$ being pairwise disjoint.%
 We then have,
$(\mathcal{V}_{SP}\leftrightarrow \mathcal{V}_{PQ})^\perp 
\ \equaln\ \mathcal{V}_{SP}^\perp \rightleftharpoons \mathcal{V}_{PQ}^\perp 
.$ In particular,
$(\mathcal{V}_{SP}\leftrightarrow \mathcal{V}_{P})^\perp \ \equaln\ \mathcal{V}_{SP}^\perp \leftrightarrow \mathcal{V}_{P}^\perp
.$
\end{theorem}
A proof of a general version of Theorem \ref{thm:idt0}, is given in the
appendix.

We present some important applications for multiports and multiport behaviours
in the form of  corollaries. These have the form `if the device characteristic
of a multiport has a certain property, so does the port behaviour'.
The implicit duality theorem is useful in proving such results 
provided the property is in some way related to orthogonality.

First we need the definitions of  some special characteristics and relationships between them.

Let $S$ denote the set of edges of the graph of the linear network and let
$\{S_1, \cdots , S_k\} $ be a partition of $S,$ each block $S_j$ being an
\nw{individual device}. Let $S',S"$ be copies of  $S,$
with $e',e"$ corresponding to edge $e.$ The device characteristic would usually have the form $\bigoplus \A_{S_j'S_j"},$
defined by $(B_j'v_{S_j'}+Q_j"i_{S_j"})=s_j,$
with rows of $(B_j'|Q_j")$ being linearly independent.
 The vector space translate
would have the form $\bigoplus \V_{S_j'S_j"},$
$\V_{S_j'S_j"}$ being the translate of $\A_{S_j'S_j"}.$
Further, usually the blocks $S_j$ would have size not more than ten, even
when the network has millions of nodes.
Therefore, it would be easy to build $(\bigoplus \V_{S_j'S_j"})^{\perp}=\bigoplus \V_{S_j'S_j"}^{\perp}.$


We say $\V_{S_j'S_j"},\hat{\V}_{S_j'S_j"}$ are \nw{orthogonal duals}
of each other iff $\hat{\V}_{S_j'S_j"}= \V^{\perp}_{S_j'S_j"}$
and denote the orthogonal duals of $\V_{S_j'S_j"}$ by $\V^{dual}_{S_j'S_j"}.$
It is clear that $(\V^{dual}_{S_j'S_j"})^{dual}=\V_{S_j'S_j"}.$\\
Let $\K_{S_j'S_j"}$ be an affine space with $\V_{S_j'S_j"}$ as its 
vector space translate.
We say $\hat{\V}_{S_j'S_j"}$ is the  \nw{adjoint}
of $\K_{S_j'S_j"},$ 
denoted by ${\V}^{adj}_{S_j'S_j"},$
iff $\hat{\V}_{S_j'S_j"}= (\V^{\perp}_{S_j'S_j"})_{(-S_j")S_j'}.$
It is clear that 
${\V}_{S_j'S_j"}= (\hat{\V}^{\perp}_{S_j'S_j"})_{(-S_j")S_j'}.$

We say $S_j$ is  \nw{reciprocal}
iff $\K_{S_j'S_j"}\equivd \V_{S_j'S_j"},$ where
$(\V_{S_j'S_j"})_{-S_j"S_j'}=\V_{S_j'S_j"}^{\perp},\ equivalently,$ \\
$(\V_{S_j'S_j"}^{\perp})_{(-S_j")S_j'}=\V_{S_j'S_j"}.$\\
Thus $\V_{S_j'S_j"}$ is reciprocal iff it is self-adjoint.\\
Note that  if $\V_{S_j'S_j"}$ is defined by $v_{S_j'}=Ki_{S_j"},
\V_{S_j'S_j"}^{\perp}$ would be defined by $i_{S_j"}=-K^Tv_{S_j'}$
and \\ $(\V_{S_j'S_j"}^{\perp})_{(-S_j")S_j'}$ would be defined by
$v_{S_j'}=K^Ti_{S_j"},$
so that $\V_{S_j'S_j"},$ being reciprocal
means that $K=K^T.$\\
We say $\V_{S_j'S_j"},\hat{\V}_{S_j'S_j"}$ are \nw{Dirac duals}
of each other iff $\hat{\V}_{S_j'S_j"}= (\V^{\perp}_{S_j'S_j"})_{S_j"S_j'},$
i.e., iff \\
${\V}_{S_j'S_j"}= (\hat{\V}^{\perp}_{S_j'S_j"})_{S_j"S_j'}$
and denote the Dirac dual   of $\V_{S_j'S_j"}$ by $\V^{Ddual}_{S_j'S_j"}.$
It is clear that $(\V^{Ddual}_{S_j'S_j"})^{Ddual}=\V_{S_j'S_j"}.$\\
We say $\K_{S_j'S_j"}$ is  \nw{Dirac}
iff $\K_{S_j'S_j"}\equivd \V_{S_j'S_j"},$ where
$(\V_{S_j'S_j"})_{S_j"S_j'}=\V_{S_j'S_j"}^{\perp}.$
\\
Thus $\V_{S_j'S_j"}$ is Dirac iff it is self-Dirac dual.
In this case, if $\V_{S_j'S_j"}$ is defined by  $v_{S_j'}=Ki_{S_j"},\\
(\V_{S_j'S_j"})_{S_j"S_j'}=\V_{S_j'S_j"}^{\perp}$ would be defined by
$v_{S_j'}=-K^Ti_{S_j"},$ so that $K=-K^T.$\\
\begin{example}
\label{eg:adjointetc}
1. Let $v= Ri+\E$ be a source accompanied resistor, $R$ being diagonal,
with nonzero diagonal entries.
The vector space translate of the characteristic is the solution space 
of $v= Ri.$ The orthogonal dual has the characteristic $\hat{i}=-R^T\hat{v}=-R\hat{v}.$
The adjoint has the characteristic $\tilde{v}=R\tilde{i}.$
Thus the source free resistor is self adjoint or reciprocal.
This would be true even if $R$ is a symmetric matrix of full rank.\\
2. A device with the hybrid characteristic and its adjoint are shown below.
\begin{align}
\label{eqn:hybrid}
\ppmatrix{v_1\\i_2}=\ppmatrix{r_{11}& h_{12}\\
h_{21}& g_{22}} \ppmatrix{i_1\\v_2}; \ \ \ \ \ \ \ \ \ \  \ 
\ppmatrix{\hat{v}_1\\\hat{i}_2}=\ppmatrix{r^T_{11}& -h^T_{21}\\
-h^T_{12}& g^T_{22}} \ppmatrix{\hat{i}_1\\\hat{v}_2};
\end{align}
3. The Dirac dual of $v= Ri+\E$ has the characteristic $\tilde{v}=-R^T\tilde{i}.$
Therefore the source free device $v= Ri$ would be self Dirac dual iff
$R=-R^T.$
\end{example}

We say $S_j$ is an \nw{ideal transformer}
iff $\K_{S_j'S_j"}\equivd \V_{S_j'}\oplus (\V_{S_j'}^{\perp})_{S_j"}.$\\
Let
$\tilde{P}$ be a disjoint copy of $P,$
with $P\equivd\{e_1, \cdots , e_k\}, \tilde{P}\equivd \{\tilde{e}_1, \cdots , \tilde{e}_k\}, e_i,\tilde{e}_i$ being copies of each other.
We note that $\T^{P\tilde{P}},$ defined in Subsection \ref{subsec:connectioncomposition} denotes the `$(1:1)$ ideal transformer'
defined by
$v_{e_i}=v_{\tilde{e}_i}, i_{e_i}=-i_{\tilde{e}_i}, i=1, \cdots ,k.$
\\
Since $(\V_{S_j'}\oplus (\V_{S_j'}^{\perp})_{S_j"})^{adj}=
((\V_{S_j'}\oplus (\V_{S_j'}^{\perp})_{S_j"})^{\perp})_{(-S_j")S_j'}
=\V_{S_j'}\oplus (\V_{S_j'}^{\perp})_{S_j"},$
an ideal transformer is self-adjoint, i.e., reciprocal.

Let $S_j, \hat{S}_j$ be pairwise disjoint copies with $S_j\equivd S_j\uplus \hat{S}_j, S_j\equivd \{e_1,\cdots e_k\}, \hat{S}_j\equivd  \{\hat{e}_1, \cdots \hat{e}_k\}.$
We say $\V_{S_j'\hat{S}_j'S_j"\hat{S}_j"}$ is a \nw{gyrator} 
iff
$v_{S_j'}=-R i_{\hat{S}_j"}; v_{\hat{S}_j'}=R i_{S_j"},$ where $R$  is a positive diagonal matrix. It is clear that a gyrator is
Dirac.
\\
We denote by $\g ^{S_j\hat{S}_j},$ the gyrator where $R$ is the identity matrix.
\\
We say $\V_{S_j'S_j"}$ is \nw{passive} iff, whenever $(v_{S_j'},i_{S_j"})\in \V_{S_j'S_j"},$
we  have that the dot product $\langle v_{S_j'}, i_{S_j"} \rangle$ is nonnegative.\\
As mentioned before, we say $S_j$ is  a set of \nw{norators}
iff $\K_{S_j'S_j"}\equivd \F_{S_j'S_j"},$ i.e., there are no constraints on
$v_{S_j'},i_{S_j"}.$
 We say $e_j$ is a \nw{nullator} iff
$v_{e_j'}=i_{e_j"}=0.$
Note that $e_j$ is a {nullator} iff its orthogonal dual/adjoint/Dirac dual is a norator.

\begin{corollary}
\label{cor:idealtransformer}
Let $\V_{S'P'S"P"},\V_{S'S"},$ be ideal transformers. Then so is
$\V_{S'P'S"P"}\lrar \V_{S'S"}. $
\end{corollary}
\begin{proof}
We have $\V_{S'P'S"P"}\equivd \V_{S'P'}\oplus \V^{\perp}_{S"P"},
\V_{S'S"}\equivd \V_{S'}\oplus \V^{\perp}_{S"},$ so that\\
$\V_{S'P'S"P"}\lrar \V_{S'S"}= (\V_{S'P'}\oplus \V^{\perp}_{S"P"})\lrar (\V_{S'}\oplus \V^{\perp}_{S"})= (\V_{S'P'}\lrar \V_{S'})\oplus (\V^{\perp}_{S"P"}\lrar \V^{\perp}_{S"}).$ Since, by Theorem \ref{thm:idt0}, 
$(\V_{S'P'}\lrar \V_{S'})^{\perp}=(\V^{\perp}_{S"P"}\lrar\V^{\perp}_{S"}),$
the result follows.
\end{proof}

\begin{corollary}
\label{cor:reciprocalDirac}
\begin{enumerate}
\item Let $\V_{S'P'S"P"}, \hat{\V}_{S'P'S"P"}$ be adjoints of each other
and so also $\V_{S'S"}, \hat{\V}_{S'S"}.$
Then $\V_{S'P'S"P"}\lrar \V_{S'S"}, $ and $\hat{\V}_{S'P'S"P"}\lrar \hat{\V}_{S'S"}, $ are also adjoints of each other.
\item Let $\V_{S'P'S"P"}, \hat{\V}_{S'P'S"P"}$ be Dirac duals of each other
and so also $\V_{S'S"}, \hat{\V}_{S'S"}.$\\
Then $\V_{S'P'S"P"}\lrar \V_{S'S"}, $ and $\hat{\V}_{S'P'S"P"}\lrar \hat{\V}_{S'S"}, $ are also Dirac duals  of each other.
\item Let $\V_{S'P'S"P"},\V_{S'S"},$ be reciprocal. Then so is 
$\V_{S'P'S"P"}\lrar \V_{S'S"}. $ 
\item Let $\V_{S'P'S"P"},\V_{S'S"},$ be Dirac. Then so is
$\V_{S'P'S"P"}\lrar \V_{S'S"}. $
\end{enumerate}
\end{corollary}
\begin{proof}
\begin{enumerate}
\item Let $ \V_{P'P"}\equivd \V_{S'P'S"P"}\lrar \V_{S'S"}$
and let $ \hat{\V}_{P'P"}\equivd \hat{\V}_{S'P'S"P"}\lrar \hat{\V}_{S'S"}.$

We then have,\\
$ (\V_{P'P"}^{\perp})_{(-P")P'}= ((\V_{S'P'S"P"}\lrar \V_{S'S"})^{\perp})_{(-P")P'}= (\V_{S'P'S"P"}^{\perp}\lrar \V_{S'S"}^{\perp})_{(-P")P'}$\\$=(\V_{S'P'S"P"}^{\perp})_{-S"(-P")S'P'}\lrar (\V_{S'S"}^{\perp})_{-S"S'}= \hat{\V}_{S'P'S"P"}\lrar \hat{\V}_{S'S"}=\hat{\V}_{P'P"} .$
\item The proof is essentially the same as the previous part except 
that the subscripts $S",P"$ do not have a negative sign.
\item 
Let $\V_{P'P"}\equivd (\V_{S'P'S"P"}\lrar \V_{S'S"}).$
We have  $ (\V_{P'P"}^{\perp})_{(-P")P'}= ((\V_{S'P'S"P"}\lrar \V_{S'S"})^{\perp})_{(-P")P'}=(\V_{S'P'S"P"}^{\perp}\lrar \V_{S'S"}^{\perp})_{(-P")P'}=  (\V_{S'P'S"P"}^{\perp})_{-S"(-P")S'P'}\lrar (\V_{S'S"}^{\perp})_{-S"S'}
={\V}_{S'P'S"P"}\lrar {\V}_{S'S"}={\V}_{P'P"} .$
\item  The proof is essentially the same as the previous part except
that the subscripts $S",P"$ do not have a negative sign.
\end{enumerate}

\end{proof}

Let the multiport $\mnw{{\N}_P}$ be on graph $\G_{SP}$
with device characteristic
$\K_{S'S"}=x_{S'S"}+{\V}^{}_{S'S"}$ on $S.$\\
The  multiport $\mnw{{\N}^{hom}_P}$ is on graph $\G_{SP}$
but has device characteristic ${\V}^{}_{S'S"}.$\\
The \nw{adjoint} $\N^{adj}_P$ of ${\N}_P$ as well as of ${\N}^{hom}_P$ 
is on graph $\G_{SP}$
but has device characteristic ${\V}^{adj}_{S'S"}.$
%
%

Let the solution set of $\N_P$
be
$[\V_{S'P'}\oplus (\V^{\perp}_{S'P'})_{S"P"}]\cap \K_{S'S"}.$
Then \\the solution set of  ${\N}^{hom}_P$
would be
$[\V_{S'P'}\oplus (\V^{\perp}_{S'P'})_{S"P"}]\cap {\V}^{}_{S'S"},
$
and \\that of  $\N^{adj}_P$
would be
$[\V_{S'P'}\oplus (\V^{\perp}_{S'P'})_{S"P"}]\cap {\V}^{adj}_{S'S"}.$

The port behaviour of $\N_P$
would be $[\V_{S'P'}\oplus (\V^{\perp}_{S'P'})_{S"-P"}]\lrar \K_{S'S"},$\\
that of 
${\N}^{hom}_P$
would be
$[\V_{S'P'}\oplus (\V^{\perp}_{S'P'})_{S"-P"}]\lrar {\V}^{}_{S'S"},
$
and\\ that of  $\N^{adj}_P$
would be
$[\V_{S'P'}\oplus (\V^{\perp}_{S'P'})_{S"-P"}]\lrar {\V}^{adj}_{S'S"}.$


We now have the following basic result on linear multiports.
\begin{theorem}
\label{thm:adjointmultiport}
\begin{enumerate}
\item Let $\N^1_P,{\N}^2_P$ be multiports on the same graph $\G_{SP}$
but with device characteristics \\
$\V^1_{S'S"}, {\V}^2_{S'S"},$ respectively
and port behaviours $\breve{\V}^1_{P'P"}, \breve{\V}^2_{P'P"},$ respectively.
Then if $\V^1_{S'S"}, {\V}^2_{S'S"},$ are adjoints (Dirac duals) of each other
so are $\breve{\V}^1_{P'P"}, \breve{\V}^2_{P'P"},$ adjoints (Dirac duals) of each other.
\item If the port behaviour $\breve{\V}_{P'P"}$ of $\N_P$ is proper,
then so is the port behaviour $(\breve{\V}_{P'P"})^{adj}$ of $\N^{adj}_P.$
\item If $\N_P$ has a device characteristic $\V_{S'S"}$ that is reciprocal
(self-Dirac dual) then its port behaviour $\breve{\V}_{P'P"}$ is also
reciprocal
(self-Dirac dual). Further the port behaviour has dimension $|P|.$
\item If $\N_P$ has a device characteristic $\V_{S'S"}$ that is an ideal transformer,
then its port behaviour $\breve{\V}_{P'P"}$ is also
an ideal transformer.
\end{enumerate}
\end{theorem}
\begin{proof}
1. Let $\V_{S'P'}\equivd (\V^v(\G_{SP}))_{S'P'}.$
By Tellegen's Theorem, we have $(\V^v(\G_{SP}))^{\perp}=\V^i(\G_{SP}).$\\
Therefore $((\V_{S'P'})^{\perp})_{S"P"}=(\V^i(\G_{SP}))_{S"P"}.$
Let $\V^1_{P'P"}\equivd(\V_{S'P'}\oplus ((\V_{S'P'})^{\perp})_{S"P"})\lrar \V^1_{S'S"}.$ and let ${\V}^2_{P'P"}=(\V_{S'P'}\oplus ((\V_{S'P'})^{\perp})_{S"P"})\lrar {\V}^2_{S'S"}.$
It is clear that $\V_{S'P'}\oplus ((\V_{S'P'})^{\perp})_{S"P"}$
is self-adjoint. We are given that $\V^1_{S'S"},{\V}^2_{S'S"}$ are adjoints.
Therefore, by part 1 of Corollary \ref{cor:reciprocalDirac}
, we have that\\ $\V^1_{P'P'},{\V}^2_{P'P'}$ are adjoints and therefore
$\breve{\V}^1_{P'P'}\equivd (\V^1_{P'P'})_{P'(-P")}$ and
$\breve{\V}^2_{P'P'}\equivd ({\V}^2_{P'P'})_{P'(-P")}$ are adjoints.

 The Dirac dual case is similar to the above noting that $\V_{S'P'}\oplus ((\V_{S'P'})^{\perp})_{S"P"}$
is self-Dirac dual.\\
2. Since $r((\breve{\V}^{\perp}_{P'P'}))=2|P|-r(\breve{\V}_{P'P'})=|P|,$
we also have $r((\breve{\V}_{P'P'})^{adj})=r((\breve{\V}_{P'P'})_{(-P")P'})=|P|.$\\
3. The reciprocal case follows from part 1 above since $\V_{S'S"}$ is given to be reciprocal, i.e., self-adjoint.\\ By the definition of self-adjointness, we must have
that $\V_{P'P"}$ and $\V^{\perp}_{P'P"}$ must have the same dimension.
But their dimensions must add upto $2|P|.$ Therefore $\V_{P'P'}$
 and \\$\breve{\V}_{P'P'}\equivd (\V_{P'P'})_{P'(-P")}$
have dimension $|P|.$
\\
 The self-Dirac dual case follows from part 1 above since $\V_{S'S"}$ is given to be
 self-Dirac dual.\\
4. We note that $\V_{S'P'}\oplus ((\V_{S'P'})^{\perp})_{S"P"}$
is an ideal transformer. We are given that $\V_{S'S"}$ is an ideal transformer.
The result now follows from Corollary \ref{cor:idealtransformer}.

\end{proof}

\section{Explicit computation of  solution and port behaviour of multiport networks}
\label{sec:linearalgebrabehaviour}
In this section we consider the linear algebraic problem of
computing the solution and port behaviour of linear multiport networks.
This is a problem of elimination of variables by suitable 
row operations on the governing equations of the linear multiport.
Our treatment is general, in the sense that it will work 
for arbitrary linear multiports, even if they have no solution.
However, treating this as a general linear algebraic problem is not as
convenient as reducing it to repeated solution of circuits
with unique solution, because versatile circuit simulators are available
for the latter purpose. We perform this reduction in Section \ref{sec:computingbehaviour}
for the important class of multiports which have nonvoid sets of solutions 
 for arbitrary source values 
and further, whose port conditions uniquely fix internal conditions (`regular'
 multiports in Definition \ref{def:regular}).
\subsection{Equations of  multiports }
We now study the constraints of mutiports and port behaviours
explicitly in terms of their governing equations
using the ideas of the Subsection \ref{subsec:elimination}.
Let multiport $\mathcal{N}_P $ be on graph $\G_{SP}$ with affine device characteristic $\A_{S'S"}.$ 
Let $\V_{S'P'}\equivd (\V^v(\G_{SP}))_{S'P'},$ so that $ (\V^{\perp}_{S'P'})_{S"P"}= (\V^i(\G_{SP}))_{S"P"}.$ 
Let us  compute the port behaviour $\breve{\A}_{P'P"}$ of $\N_P.$

Let $(Q_{S}|Q_{P}),(B_{S}|B_{P})$ be representative matrices of
$(\V^v(\G_{SP})),(\V^i(\G_{SP})),$ respectively and 
let $\A_{S'S"}$ be the affine space that is the set of solutions
of $(M_{S'})v_{S'}+(N_{S"})i_{S"}=s_D.$
We will assume that the rows of $(M_{S'}|N_{S"}),$ are linearly independent.
Let $\V_{S'S"}$ be the vector space translate of $\A_{S'S"}.$
It is clear that $\V_{S'S"}^{\perp}$ is the row space of $(M_{S'}|N_{S"}).$

Then the set of solutions of $\mathcal{N}_P,$ 
given by $[\V_{S'P'}\oplus (\V^{\perp}_{S'P'})_{S"P"}]\cap \A_{S'S"},$
can be cast as  the set of solutions of the equation
\begin{align}
\label{eqn:elimination4}
\ppmatrix{
        B_{S'} & \vdots\vdots  & B_{P'} &\vdots\vdots  & 0_{S"} & \vdots\vdots  & 0_{P"}\\
        0_{S'} & \vdots\vdots  & 0_{P'}& \vdots\vdots  & Q_{S"} & \vdots\vdots  & Q_{P"}\\
        M_{S'} & \vdots\vdots  & 0_{3P'}& \vdots\vdots  & N_{S"} & \vdots\vdots  & 0_{3P"}}
        \ppmatrix{v_{S'}\\v_{P'}\\i_{S"}\\i_{P"}}&=\ppmatrix{0 \\ 0\\ s_D }. 
\end{align}
As  in Equation \ref{eqn:elimination2}, we can perform invertible row
operations on  Equation \ref{eqn:elimination4}, to yield 
\begin{align}
\label{eqn:elimination5}
\ppmatrix{
        B_{1S'} & \vdots\vdots  & B_{1P'} &\vdots\vdots  & Q_{1S"} & \vdots\vdots  & Q_{1P"}\\
        M_{S'} & \vdots\vdots  & 0_{2P'}& \vdots\vdots  & N_{S"} & \vdots\vdots  & 0_{2P"}\\
        0_{3S'} & \vdots\vdots  & B_{3P'}& \vdots\vdots  & 0_{3S"} & \vdots\vdots  & Q_{3P"}}
        \ppmatrix{v_{S'}\\v_{P'}\\i_{S"}\\i_{P"}}&=\ppmatrix{{s}_1\\ {s}_D \\ {s}_3}, 
\end{align}
where the rows of 
\begin{align}
\label{eqn:reducedcoeff}
\ppmatrix{ B_{1S'} & \vdots\vdots    & Q_{1S"} \\
        \M_{S'} & \vdots\vdots  &  N_{S"} }
\end{align}
are linearly independent.
The port behaviour $\breve{\A}_{P'P"}$ of $\N_P$ is given by\\
$[((\V_{S'P'}\oplus (\V^{\perp}_{S'P'})_{S"P"})\cap \A_{S'S"})\circ P'P"]_{P'(-P")}.$\\
Therefore, $\breve{\A}_{P'P"}$ is  the set of solutions of
$B_{3P'}v_{P'}+Q_{3P"}(-i_{P"})=s_3.$

\begin{remark}
\label{rem:elimination}
We note the following in regard to the consistency of the multiport $\N_P,$ to the nature of $\breve{\A}_{P'P"}$
and its relation to the `interior' variables (on $S',S"$) of the multiport. 
\begin{enumerate}
\item The multiport $\N_P$ is consistent for arbirary source values (of the device characteristic), i.e., Equation \ref{eqn:elimination4}
has a solution for arbitrary values of the vector $s_D,$ 
iff row space of the first two rows has zero
intersection with the row space of the last row of the coefficient matrix,
\\ i.e., iff 
$[(\V^i(\G_{SP}))_{S'P'}\oplus  (\V^v(\G_{SP}))_{S"P"}] \cap[\0_{P'P"}\oplus  \V_{S'S"}^{\perp}]
=\0_{P'P"S'S"},$
\\ i.e., iff 
$[(\V^i(\G_{SP}))_{S'P'}\oplus  (\V^v(\G_{SP}))_{S"P"}]\times S'S" \cap \V_{S'S"}^{\perp}
=\0_{S'S"}.$
\\ Since the first two rows and the third row of Equation \ref{eqn:elimination4}
are given to be individually linearly independent, this implies that the 
set of rows of the 
coefficient matrix is linearly independent.
\item 
Equation \ref{eqn:elimination5} is obtained by invertible row
transformation of Equation \ref{eqn:elimination4}. Therefore if $\N_P$ is consistent for arbirary source values,
the matrix $(B_{3P'}|Q_{3P"})$ has linearly independent rows
and the port behaviour is nonvoid for arbirary source values.
\item Next let us examine when for a given $(v_{P'},-i_{P"}) \in \breve{\A}_{P'P"},$
there is a unique vector $(v_{S'},v_{P'},i_{S"},i_{P"})$
that is a solution to the multiport $\N_P.$
This happens provided the matrix in Equation \ref{eqn:reducedcoeff} is nonsingular. This is equivalent to the columns corresponding to $S',S"$ being 
linearly independent in the coefficient matrix of Equation \ref{eqn:elimination4}, i.e.,  equivalent to the restriction to set $S'\uplus S",$ of the row space of the coefficient 
matrix of Equation \ref{eqn:elimination4}, being equal to the full space $\F_{S'S"},$
\\i.e., $[(\V^i(\G_{SP}))_{S'P'}\oplus  (\V^v(\G_{SP}))_{S"P"}]\circ S'S"+\V_{S'S"}^{\perp}
=\F_{S'S"}.$
Note that consistency of the multiport does not imply this property.

\item The vector space translate of the port behaviour 
can have dimension ranging from $0$ to $2|P|,$
zero corresponding to $P$ being a set of norators and $2|P|,$ to
$P$ being a set of nullators in $\breve{\A}_{P'P"}$ 
 (see \cite{recski19}).
\end{enumerate}
\end{remark}
\begin{remark}
\label{rem:elimination0}
Consider the situation when $\N_P$ has no ports, i.e., when $P=\emptyset.$
We assume $B_{S'},Q_{S"}$ are the representative matrices
of $(\V^{i}(\G_S))_{S'},(\V^{v}(\G_S))_{S"}$ respectively.
In this case the network has a 
solution for arbitrary source values of the device characteristic 
iff $[(\V^i(\G_{S}))_{S'}\oplus  (\V^v(\G_{S}))_{S"}] \cap \V_{S'S"}^{\perp}
=\0_{S'S"}.$
Given that it has a solution, it has a 
unique solution iff $[(\V^i(\G_{S}))_{S'}\oplus  (\V^v(\G_{S}))_{S"}]+\V_{S'S"}^{\perp}
=\F_{S'S"}.$
\end{remark}
\section{Regular  Multiports}
\label{subsec:regular}
A multiport  that satisfies the properties in parts 1 to 3 of Remark \ref{rem:elimination},
can be handled by freely
available circuit simulators after some preprocessing as shown in Section \ref{sec:computingbehaviour}.
We therefore give a name to such multiports and discuss their properties.
\begin{definition}
\label{def:regular}
Let multiport $\N_P$ be on graph $\G_{SP}$ and device characteristic 
 $\A_{S'S"}
=\alpha_{S'S"}+\V_{S'S"}.$\\
The multiport $\N_P$ is said to be \nw{regular} iff every  multiport $\hat{\N}_P$
on graph $\G_{SP}$ and device characteristic
 $\hat{\A}_{S'S"}=\hat{\alpha}_{S'S"}+\V_{S'S"},$
has a non void set of
solutions  
and has a unique solution corresponding to every vector in its port behaviour. A regular multiport is said to be \nw{proper} iff its port 
behaviour is proper.
\end{definition}

We  restate the first three parts of Remark \ref{rem:elimination} in a more convenient form and derive consequences below.
\begin{theorem}
\label{thm:regularmultiportconditions}
Let multiport $\N_P$ be on graph $\G_{SP}$ with device characteristic $\A_{S'S"}=\alpha_{S'S"}+\V_{S'S"}.$ Let $\N^{\beta}_P$ denote a multiport on graph $\G_{SP}$ with device characteristic $\beta_{S'S"}+\V_{S'S"}.$ 	
Then the following hold. 
\begin{enumerate}
\item The set of solutions of $\N^{\beta}_P$ for every vector $\beta_{S'S"}$ is non void \\ iff  $\V^v(\G_{SP}\circ{S})_{S'}\oplus \V^i(\G_{SP}\times S)_{S"}+\V_{S'S"}=\F_{S'S"}.$ 
\\
This implies that  the port behaviour  of every $\N^{\beta}_P$ is non void.
\item For each $(v_{P'},-i_{P"})\in \breve{\A}_{P'P"},$ where 
$\breve{\A}_{P'P"}$ is the port behaviour of $\N_P,$
there is a unique solution $(v_{S'},v_{P'},i_{S"},i_{P"})$
of $\N_P$
iff $[\V^v(\G_{SP}\times{S})_{S'}\oplus \V^i(\G_{SP}\circ S)_{S"}]\cap \V_{S'S"}=\0_{S'S"}.$
\item (a) $\N_P$ is regular iff it satisfies the conditions in the previous two parts. \\
(b) If $\N_P$ is regular, $r(\V_{S'S"})+r(\breve{\V}_{P'P"})=|S|+|P|.$\\
(c) If $\N_P$ is regular and the device characteristic $\A_{S'S"}$ is proper,
then $\N_P$ is  proper.
\item (a) $\N_P^{adj}$ is regular iff $\N_P$ is.\\
(b)  $\N_P^{adj}$ is proper iff $\N_P$ is.
\end{enumerate}
\end{theorem}
\begin{proof}
Let us denote
$((\V^v(\G_{SP}))_{S'P'} \oplus (\V^i(\G_{SP}))_{S"P"})$
by ${\V}_{S'P'S"P"}.$
We then have,\\
${\V}_{S'P'S"P"}^{\perp}= (\V^v(\G_{SP}))_{S'P'}^{\perp} \oplus (\V^i(\G_{SP}))^{\perp}_{S"P"}=(\V^i(\G_{SP}))_{S'P'} \oplus (\V^v(\G_{SP}))_{S"P"},$
(using Theorem \ref{thm:tellegen}).\\
By Theorems \ref{thm:sumintersection} and
 \ref{thm:dotcrossidentity}, the condition ${\V}_{S'P'S"P"}^{\perp}\times S'S" \cap \V_{S'S"}^{\perp}=\0_{S'S"}$ of part 1 of Remark \ref{rem:elimination}
\\ is equivalent to
${\V}_{S'P'S"P"}\circ S'S" + \V_{S'S"}=\F_{S'S"}$\\
and  the condition ${\V}_{S'P'S"P"}^{\perp}\circ S'S" + \V_{S'S"}^{\perp}=\F_{S'S"}$ of part 3 of Remark \ref{rem:elimination}\\
is equivalent to
${\V}_{S'P'S"P"}\times S'S" \cap  \V_{S'S"}=\0_{S'S"}.$\\
Next by Lemma \ref{lem:minorgraphvectorspace},\\ ${\V}_{S'P'S"P"}\circ S'S"=(\V^v(\G_{SP}))_{S'P'} \circ S'\oplus (\V^i(\G_{SP}))_{S"P"}\circ S"= 
(\V^v(\G_{SP}\circ{S}))_{S'}\oplus (\V^i(\G_{SP}\times S))_{S"}$\\
and ${\V}_{S'P'S"P"}\times S'S"=(\V^v(\G_{SP}))_{S'P'} \times S'\oplus (\V^i(\G_{SP}))_{S"P"}\times S"= 
(\V^v(\G_{SP}\times{S}))_{S'}\oplus (\V^i(\G_{SP}\circ S))_{S"}.$
Thus parts 1 and 2 of Remark \ref{rem:elimination} reduce to part 1 of the present
lemma and part 3 of Remark \ref{rem:elimination} reduces to part 2 of the present
lemma. \\Part 3(a) of the theorem follows immediately.\\
3(b) If $r(\V_{S'S"})=k,$  taking the 
coefficient matrix $(M_{S'}|N_{S"})$ in Equation \ref{eqn:elimination4} to have linearly 
independent rows, since $r(\V^v(\G_{SP}))+r(\V^i(\G_{SP})) = |S|+|P|,$ we see that Equation \ref{eqn:elimination4} has $|S|+|P|+k$ rows.
Regularity implies that these rows are linearly independent and also
that the coefficient matrix in Equation \ref{eqn:reducedcoeff} is nonsingular.
Since the number of columns is $|S|+|P|+k,$
and also
 the coefficient matrix in Equation \ref{eqn:reducedcoeff} is nonsingular, it follows  
that the third set of rows of Equation \ref{eqn:elimination5}, i.e., $(B_{3P'}|Q_{3P"})$
is linearly independent and $|P|+k-|S|$ in number. 
This means that $r((\breve{\V}_{P'P"})_{P'(-P")})=2|P|-(|P|+k-|S|)=|P|+|S|-k,$
so that $r(\breve{\V}_{P'P"})=|P|+|S|-k$ and $r(\V_{S'S"})+r(\breve{\V}_{P'P"})=|P|+|S|.$
\\
3(c) Follows by setting $k=|S|.$
\\ 
4(a) 
 We saw above that the condition $[\V^v(\G_{SP}\times{S})_{S'}\oplus \V^i(\G_{SP}\circ S)_{S"}]
\cap \V_{S'S"}=\0_{S'S"}$
is equivalent to ${\V}_{S'P'S"P"}^{\perp}\circ S'S" + \V_{S'S"}^{\perp}=\F_{S'S"},$
i.e., equivalent to $({\V}_{S'P'S"P"}^{\perp}\circ S'S")_{(-S")S'} + (\V_{S'S"}^{\perp})_{(-S")S'}=\F_{S'S"},$\\
i.e., equivalent to $({\V}_{S'P'S"P"}\times S'S")^{\perp}_{(-S")S'} + (\V_{S'S"}^{\perp})_{(-S")S'}=\F_{S'S"},$
\\i.e., equivalent to $(\V^v(\G_{SP}\circ{S}))_{S'}\oplus (\V^i(\G_{SP}\times S))_{S"}
+\V_{S'S"}^{adj}=\F_{S'S"}.$\\
Similarly, $(\V^v(\G_{SP}\circ{S}))_{S'}\oplus (\V^i(\G_{SP}\times S))_{S"}
+ \V_{S'S"}=\F_{S'S"}$\\
can be shown to be equivalent to 
$[(\V^v(\G_{SP}\times{S}))_{S'}\oplus (\V^i(\G_{SP}\circ S))_{S"}]
\cap \V_{S'S"}^{adj}=\0_{S'S"}.$\\
Since $\N_P^{adj}$ is on graph $\G_{SP}$ and has device characteristic 
$\V_{S'S"}^{adj},$ the result follows from the previous two parts.
\\
4(b) This follows from the fact that if a vector space $\breve{\V}_{P'P"}$
has dimension $|P|,$ its complementary orthogonal space and therefore, its adjoint, 
will also have dimension $|P|.$
\end{proof}

When we set $P=\emptyset $ in the above lemma, we get the following corollary.
\begin{corollary}
\label{cor:regularmultiportconditions0}
Let network $\N$ be on graph $\G_{S}$ with device characteristic $\A_{S'S"}=\alpha_{S'S"}+\V_{S'S"}.$ Let $\N^{\beta}$ denote a network on graph $\G_{S}$ with device characteristic $\beta_{S'S"}+\V_{S'S"}.$
Then the following hold.
\begin{enumerate}
\item The set of solutions of $\N^{\beta}$ for every vector $\beta_{S'S"}$ is non void  iff  $\V^v(\G_{S})_{S'}\oplus \V^i(\G_{S})_{S"}+\V_{S'S"}=\F_{S'S"}.$
\item 
Given that the network $\N$ has a solution, it is unique  
iff $[\V^v(\G_{S})_{S'}\oplus \V^i(\G_{S})_{S"}]\cap \V_{S'S"}=\0_{S'S"}.$
\item Let the device characteristic $\A_{S'S"}$ be proper. Let $\N^{hom}$ denote the network on graph $\G_{S}$ with device characteristic $\V_{S'S"}.$ Then $\N$  has a unique solution iff $\N^{hom}$ has a unique 
solution.
\end{enumerate}
\end{corollary}
\begin{proof}
Parts 1 and 2 are immediate from Theorem \ref{thm:regularmultiportconditions}.
\\3. Consider the Equation \ref{eqn:elimination4} taking $P=\emptyset.$
Without loss of generality take the rows of $(M_{S'}|N_{S"})$ to be linearly
independent and therefore, since $\A_{S'S"}$ is proper, $|S|$ in number. 
In this case, Equation \ref{eqn:elimination4} has $2|S|$ equations. 
Since the number of variables is also $2|S|,$ 
it has a unique solution iff the coefficient matrix is nonsingular.
But this is also the condition for $\N^{hom}$ to have a unique
solution.

%
%
\end{proof}
\begin{remark}
The following can be shown.\\
1. A multiport with device characteristic $v=Ri,$ $R,$ a positive 
definite matrix, is regular.\\
2. If the devices are made up of controlled sources and resistors
and the defining parameters (gains, resistances, conductances)
can be taken to be algebraically independent over rationals, 
then under simply verifiable topological conditions the multiport 
is regular.\\
\end{remark}
\begin{remark}
Even if we work with the complex field, all the results of this section 
and their proofs go through taking the dot product $\langle f_X,g_X\rangle$ to be
$\Sigma f(e)g(e),e\in X.$ Wherever the adjoint is involved,
taking $\V^{adj}_{S'S"}$ to be $(\V^{\perp}_{S'S"})_{(-S")S'}$ or
as $(\V^{*}_{S'S"})_{(-S")S'}$ would both work.
($\V^{*}_{X}\equivd \{g_{X}: \langle f_{X},g_{X} \rangle =0\},$ where we define
$\langle f_X,g_X\rangle$ to be
$\Sigma f(e)\overline{g(e)},e\in X.$)
\end{remark}

\section{Generalizing Thevenin-Norton: computing port behaviour solving special circuits}
\label{sec:computingbehaviour}
One of the basic results of linear network theory is the Thevenin-Norton
Theorem for one ports ($|P|=1$) (\cite{thevenin0,mayer,norton,desoerkuh}). This result essentially computes the port
behavior of a multiport $\N_P$ in terms of its open circuit-short circuit
characteristic and its source free characteristic at the port.
It has a routine generalization to the situation where $|P|>1,$
when the port behaviour can be captured in the `hybrid' form, i.e.,
$v_{P_1}=r_{11}i_{P_1}+h_{12}v_{P_2}+\E_{P_1};i_{P_2}=h_{21}i_{P_1}+g_{21}v_{P_2}+J_{P_2}.$
This approach, although intuitive and simple, has the 
shortcoming that the the multiport may not have a solution
when a particular port is open circuited or short circuited.
Indeed, the port behaviour of the multiport may not even have a hybrid 
representation. 

Computing the port behaviour is a linear algebraic problem as 
we discussed in Section 
\ref{sec:linearalgebrabehaviour}. It is therefore solvable by standard 
methods `efficiently' (third power algorithms in the number of edges).
However, this is not convenient for our purposes.
Our  aim is to exploit the best that currently available 
circuit simulators can achieve.  
We will assume that our linear multiport is regular, i.e., has a 
solution for arbitrary source values of the 
 device characteristic 
 and, further, has a unique solution for any given port condition 
consistent with its multiport
behaviour. 

Our approach has this in common with that of Thevenin-Norton
that we also reduce the computation of port behaviour to repeated solution of circuits
with unique solution.\\
But {\it we will only attempt to capture the port behaviour in the form
$\alpha_{P'P"}+\V_{P'P"},$ with a generating matrix for $\V_{P'P"}.$}
Additional linear algebraic operations might have to be performed 
to construct a representative matrix for $\V_{P'P"}$
from the generating matrix.
Our method will {\it always work for regular multiports}.

We will assume that we have available a \nw{`standard'} linear circuit simulator
that {\it accepts circuits with  proper device characteristics  
and unique solution}.
Both, when the network has non unique solution,
and when it is inconsistent, we would get error messages.
These conditions are satisfied by  circuit simulators that are freely available.
The port edges in multiports are  norators and can be handled by a standard
circuit simulator only after some preprocessing, which we show how to perform.


Our method is to terminate  the given linear multiport $\N_P$ by
a variation of its `orthogonal dual' (i.e., the multiport with orthogonally dual constraints)
resulting in a network $\N_{final}$ which can be processed by the circuit simulator.
Briefly, suppose the port behaviour satisfies
$Kx=s,$ where the real matrix $K$ has linearly independent rows. 
 We  first try to find one solution to this equation.
Then try to find a basis for the solution space of $Kx=0.$
Let $K^{\perp}$ be a matrix whose rows are  linearly independent
and span $(row(K))^{\perp}.$ The equation
\begin{align}
\label{eqn:uniquesol1}
\ppmatrix{K\\
K^{\perp}}\ppmatrix{
x}& =  \ppmatrix{s\\o}
\end{align}
has a nonsingular coefficient matrix (as we show below) and therefore, a unique solution.
This we take as the `particular' solution $x^{p}$ to $Kx=s.$
We then consider the equation
\begin{align}
\label{eqn:uniquesol2}
\ppmatrix{K\\
K^{\perp}}\ppmatrix{
x}& =  \ppmatrix{0\\s_2^i}
\end{align}

We take right side of Equation \ref{eqn:uniquesol2}
to be vectors $s_2^i$ in turn, in general $2|P|$ in number,
and find
a solution $x^i$ to the resulting  equation. 
The set of vectors $\{x^i, i= 1, \cdots , 2|P|\},$
would be a set of generators for the vector space translate of the multiport
behaviour.
 {\it All
this we do without explicitly computing $K$ or $K^{\perp}.$}

\begin{remark}
\label{rem:realtocomplex2}
As in Remark \ref{rem:realtocomplex}, we note that, if we wish to extend the discussion of this section and Section \ref{sec:maxpower} to the complex case, we have to intepret the dot product $\langle f_X, g_X \rangle$ of $f_X,g_X$ to be the inner product $\Sigma f(e)\overline{g(e)}, e\in X,$
where $\overline{g(e)}$ is the complex conjugate of $g(e).$ 
The transpose of a matrix $Z$ has to be changed to $Z^*,$ the conjugate
transpose of $Z.$
Further,
we must interpret $\V^{\perp}_X$\\ as $\V^*_{S'S"}\equivd \{ g_X: \langle f_X, g_X \rangle =0\},$
where $\langle f_X, g_X \rangle$ is taken to be as above.
This is essential for the coefficient matrix of Equation \ref{eqn:uniquesol2}
to be nonsingular.
Algorithm \ref{alg:TN},
 statement and proof of  Lemma \ref{lem:behaviourbasis}
go through taking 
$\V^{adj}_{S'S"}$ to be
$ (\V^*_{S'S"})_{(-S")S'}.$
\\
Subsection \ref{subsec:mptcomplex} contains a discussion of the conversion
of the preceding subsections to the complex case.
\end{remark}
\subsection{Terminating a multiport by its adjoint through a gyrator}
A useful artifice for processing a multiport through a standard 
circuit simulator, which accepts only circuits with a unique solution,
 is to terminate it appropriately so that the 
resulting network, if it has a solution, has a unique solution. This solution would also 
contain a solution to the original multiport.

We now describe this technique in detail.
\\
First note that the coefficient matrix in Equation \ref{eqn:uniquesol1}
has number of rows equal to $r(span(K))+r(span(K^{\perp})),$ which is the 
number of columns of the matrix. Next, suppose the rows are linearly 
dependent. This would imply that a nontrivial linear combination of
the rows is the zero vector, which in turn implies that a nonzero vector
lies in the intersection of complementary orthogonal real vector spaces,
i.e., a nonzero real vector is orthogonal to itself, 
a contradiction. (In the complex case $x^Tx=0$ leads to $x=0$
only if we take $x^*,$  the conjugate transpose of $x,$ in place of $x^T.$ Therefore the dot 
product must be defined to be inner product for this argument to work.)
Therefore, the coefficient matrix in Equation \ref{eqn:uniquesol1}
is nonsingular.
Thus, if a matrix has two sets of rows, 
which are representative matrices of complementary orthogonal
real vector spaces,
then it must be nonsingular.

Let the multiport behaviour $\breve{\A}_{P'P"}$ be the solution space 
of the equation $B_{P'}x_{P'}-Q_{P"}y_{P"}=s,$ with linearly independent 
rows and let $\breve{\V}_{P'P"}$ 
be the solution space of the equation $B_{P'}x_{P'}-Q_{P"}y_{P"}=0.$
 Let the dual multiport behaviour $\breve{\V}^{dual}_{P'P"}$ be the solution
space of equation $B^{\perp}_{P'}v_{P'}-Q^{\perp}_{P"}y_{P"}=0,$
where the rows of $(B^{\perp}_{P'}|-Q^{\perp}_{P"})$
form a basis for the space complementary orthogonal to the 
row space of  $(B_{P'}|-Q_{P"}).$

The constraints of the two multiport behaviours together give the following
equation.
\begin{align}
\label{eqn:primaldual}
\ppmatrix{
        B_{P'}  &\vdots &  -Q_{P"} \\
        B^{\perp}_{P'} & \vdots &       -Q^{\perp}_{P"}}\ppmatrix{x_{P'}\\y_{P"}}&=\ppmatrix{s\\ 0}. 
\end{align}
The first and second set of rows of the coefficient matrix of the above 
equation are  linearly independent and span real complementary orthogonal spaces. Therefore the coefficient
matrix is invertible and the equation has a unique solution.

Let  $\mathcal{N}_P$ be on graph $\G_{SP}$
and device characteristic ${\A}_{S'S"},$ and let it have the port behaviour 
$\breve{\A}_{P'P"}.$ Let ${\V}_{S'S"}$  be the vector space translate of 
${\A}_{S'S"}$ and $\breve{\V}_{P'P"}$  be that  of 
$\breve{\A}_{P'P"}.$ We note that\\ 
$((\V^v(\G_{SP}))_{S'P'}\oplus (\V^i(\G_{SP}))_{S"P"})\lrar {\V}_{S'S"}=
(\V^v(\G_{SP}))_{S'P'}\oplus (\V^v(\G_{SP}))^{\perp}_{S"P"})\lrar {\V}_{S'S"}=
(\breve{\V}_{P'P"})_{P'(-P")};$ 
and $
(\V^v(\G_{SP}))_{S'P'}\oplus (\V^v(\G_{SP}))^{\perp}_{S"P"})\lrar {\A}_{S'S"}= (\breve{\A}_{P'P"})_{P'(-P")},$
using Theorem \ref{thm:IIT2}.

Let $\mathcal{N}^{adj}_{\tilde P}$ be on the copy $\G_{\tilde{S}\tilde{P}}$ of $\G_{SP},$ with device characteristic $({\V}^{adj}_{S'S"})_{\tilde{S}'\tilde{S}"}.$ We know that the port behaviour of  $\mathcal{N}^{adj}_{\tilde P}$ would be $\breve{\V}^{adj}_{\tilde{P}'\tilde{P}"}\equivd (\breve{\V}^{adj}_{P'P"})_{\tilde{P}'\tilde{P}"}$ (Theorem \ref{thm:adjointmultiport}). Further $(\breve{\V}^{adj}_{\tilde{P}'\tilde{P}"})_{(-\tilde{P}")\tilde{P}'}=(\breve{\V}^{\perp}_{P'P"})_{\tilde{P}'\tilde{P}"}=(\breve{\V}^{dual}_{P'P"})_{\tilde{P}'\tilde{P}"}.$ 
Thus the constraints of $\mathcal{N}^{adj}_{\tilde P}$ together 
with the constraint $ v_{P'}=-i_{\tilde{P}"}; i_{P"}=v_{\tilde{P}'}$
(the gyrator $\g^{{P}\tilde{P}}$),
correspond to the port behaviour $\breve{\V}^{dual}_{P'P"},$
which is the solution space of the second set of equations of Equation \ref{eqn:primaldual}.

Let $\N_P$ be regular. By Theorem  \ref{thm:regularmultiportconditions},
the port behaviour $\breve{\A}_{P'P"}$ is nonvoid and $\mathcal{N}^{adj}_{\tilde P}$ is also
regular.
Let us examine $[\mathcal{N}_P\oplus \mathcal{N}^{adj}_{\tilde P}]\cap \g^{{P}\tilde{P}}.$ This is a network on graph $\G_{SP}\oplus \G_{\tilde{S}\tilde{P}}$ with device characteristic $\A_{S'S"}\oplus {\V}^{adj}_{\tilde{S}'\tilde{S}"}\oplus \g^{{P}\tilde{P}}.$ This device characteristic is proper 
because\\ $r(\V_{S'S"}\oplus {\V}^{adj}_{\tilde{S}'\tilde{S}"})= 
r(\V_{S'S"})+r(\V^{\perp}_{S'S"})=2|S|= |S|+|\tilde{S}|,$
 $r(\g^{{P}\tilde{P}})=2|P|=|P\uplus \tilde{P}|,$
so that dimension of $\V_{S'S"}\oplus {\V}^{adj}_{\tilde{S}'\tilde{S}"}\oplus \g^{{P}\tilde{P}}$  equals $|S|+|\tilde{S}|+|P\uplus \tilde{P}|.$
If we restrict a solution of this network to $P'\uplus P",$ we get $v_{P'},i_{P"},$
where $v_{P'},-i_{P"} $ is a 
 solution to Equation \ref{eqn:primaldual}. We have seen that this solution is unique.
Through the constraints $ v_{P'}=-i_{\tilde{P}"}; i_{P"}=v_{\tilde{P}'}$
of $\g^{{P}\tilde{P}},$ this fixes $v_{\tilde{P}'},i_{\tilde{P}"}$ uniquely.
This means in the multiports $\mathcal{N}_P,\mathcal{N}^{adj}_{\tilde P}$
we know the (unique) port voltages and currents. 
But since $\N_P, \mathcal{N}^{adj}_{\tilde P}$ are regular, 
 all internal voltages and currents of the 
multiports $\mathcal{N}_P,  \mathcal{N}^{adj}_{\tilde P}$ are 
uniquely fixed.
Therefore this is the unique solution of the network $[\mathcal{N}_P\oplus \mathcal{N}^{adj}_{\tilde P}]\cap \g^{{P}\tilde{P}}.$ Since this network also  has proper device characteristic,
our standard circuit simulator can process it and obtain its solution.

If the multiport $\mathcal{N}_P$ is not regular,
it may not have a solution and then the port behaviour $\breve{\A}_{P'P"}$ would be  void.
Even if the multiport has a solution, so that the behaviour $\breve{\A}_{P'P"}$ is non void,
the above general procedure of solving 
$[\mathcal{N}_P\oplus \mathcal{N}^{adj}_{\tilde P}]\cap \g^{{P}\tilde{P}},$
will yield non unique internal voltages and currents in the multiports 
$\mathcal{N}_P,\mathcal{N}^{adj}_{\tilde P}.$
Therefore, $[\mathcal{N}_P\oplus \mathcal{N}^{adj}_{\tilde P}]\cap \g^{{P}\tilde{P}}$ will have  non unique solution.
(In both the above cases our standard circuit simulator would give error
messages.)
\\
Thus $[\mathcal{N}_P\oplus \mathcal{N}^{adj}_{\tilde P}]\cap \g^{{P}\tilde{P}}$ has a  unique solution iff $\N_P$  is regular.

We have computed a single vector $x_{P'P"}^{p}\equivd (v_{P'},-i_{P"})\in \breve{\A}_{P'P"}.$
We next consider the problem of finding a generating set for the vector space translate
$\breve{\V}_{P'P"}$ of $\breve{\A}_{P'P"}.$ \\Let $ \g_{tv}^{P\tilde{P}}$ denote the 
affine space that is the solution set of the constraints\\
$v_{e_j'} = -i_{\tilde{e}_j"}, e_j\in P,j\ne t,  v_{e_t'}+1 = -i_{\tilde{e}_t"}; i_{e_j"}=v_{\tilde{e}_j'}, e_j\in P.$
\\Let $ \g_{ti}^{P\tilde{P}}$ denote the
affine space that is the solution set of the constraints\\
$v_{e_j'} = -i_{\tilde{e}_j"}, e_j\in P; i_{e_j"}=v_{\tilde{e}_j'}, e_j\in P,j\ne t, i_{e_t"}+1=v_{\tilde{e}_t'}.
 $

Now solve $[\mathcal{N}^{hom}_P\oplus \mathcal{N}^{adj}_{\tilde P}]\cap \g_{tv}^{{P}\tilde{P}}$ for each $e_t'\in P'$ and $[\mathcal{N}^{hom}_P\oplus \mathcal{N}^{adj}_{\tilde P}]\cap \g_{ti}^{{P}\tilde{P}}$ for each $e_t"\in P".$
We prove below, in Lemma \ref{lem:behaviourbasis},
 that each solution yields a vector in $\breve{\V}_{P'P"}$
and the vectors corresponding to all $e_t'\in P', e_t"\in P",$ form a generating set for $\breve{\V}_{P'P"}.$

We summarize these steps in the following Algorithm.

\begin{algorithm}
\label{alg:TN}
Input: A multiport $ \N_P$ on $\G_{SP}$ with affine
device characteristic ${\A}_{S'S"}.$\\
Output: The port behaviour $\breve{\A}_{P'P"}$ of $ \N_P$ if
$ \N_P$ is regular.\\
Otherwise a statement that $ \N_P$ is not regular.

Step 1. Build the network $\N^{large}\equivd [\mathcal{N}_P\oplus \mathcal{N}^{adj}_{\tilde P}]\cap \g^{{P}\tilde{P}}$ on graph 
  $\G_{SP}\oplus \G_{\tilde{S}\tilde{P}}$ \\with device characteristic ${\A}_{S'S"}\oplus {\V}^{adj}_{\tilde{S}'\tilde{S}"}\oplus \g^{{P}\tilde{P}},$ 
where ${\V}_{S'S"}$ is the vector space translate of ${\A}_{S'S"}$
and ${\V}^{adj}_{\tilde{S}'\tilde{S}"}\equivd (({\V}_{S'S"})^{\perp})_{-\tilde{S}"\tilde{S}'}.$
\\
Find the unique solution (if it exists) of $\N^{large}$ and restrict it to $P'\uplus P"$
to obtain\\ $(v^p_{P'},i^p_{P"}).$ The vector $(v^p_{P'},-i^p_{P"})$ belongs to $ \breve{\A}_{P'P"}.$\\
If no solution exists or if there are non unique solutions output `$\N_P$ not regular' and\\ STOP.
\\
Step 2. Let $\mathcal{N}^{hom}_P$ be  obtained by replacing 
the device characteristic ${\A}_{S'S"}$ by ${\V}_{S'S"}$ in $\mathcal{N}_P.$
\\
For $t= 1, \cdots , |P|,$ build and solve $[\mathcal{N}^{hom}_P\oplus \mathcal{N}^{adj}_{\tilde P}]\cap \g_{tv}^{{P}\tilde{P}}$ and restrict it to $P'\uplus P"$
to obtain $(v^{tv}_{P'},i^{tv}_{P"}).$ 
The vector $(v^{tv}_{P'},-i^{tv}_{P"})\in \breve{\V}_{P'P"}.$
\\
For $t= 1, \cdots , |P|,$ build and solve $[\mathcal{N}^{hom}_P\oplus \mathcal{N}^{adj}_{\tilde P}]\cap \g_{ti}^{{P}\tilde{P}}$ and restrict it to $P'\uplus P"$
to obtain $(v^{ti}_{P'},i^{ti}_{P"}).$ The vector $(v^{ti}_{P'},-i^{ti}_{P"})\in \breve{\V}_{P'P"}.$
\\
Step 3. Let $\breve{\V}_{P'P"}$ be the span of the vectors
$(v^{tv}_{P'},-i^{tv}_{P"}), (v^{ti}_{P'},-i^{ti}_{P"}), t= 1, \cdots ,|P|.$
\\
Output $\breve{\A}_{P'P"}\equivd (v^p_{P'},-i^p_{P"})+\breve{\V}_{P'P"}.$\\
STOP
\end{algorithm}
We complete the justification of Algorithm \ref{alg:TN} in the following 
lemma. We remind the reader that $\mathcal{N}^{hom}_P$ is obtained 
from $\N_{P}$ by replacing its device characteristic $\A_{S'S"}$ by
$\V_{S'S"}.$
\begin{lemma}
\label{lem:behaviourbasis}
Let $\N_{P}$ on graph $\G_{SP}$ be regular.
Then the following hold.
\begin{enumerate}
\item  The network $\N^{large}\equivd [\mathcal{N}_P\oplus \mathcal{N}^{adj}_{\tilde P}]\cap \g^{{P}\tilde{P}}$ has a proper device characteristic and has a unique solution.
\item Each of the  networks $[\mathcal{N}^{hom}_P\oplus \mathcal{N}^{adj}_{\tilde P}]\cap \g_{tv}^{{P}\tilde{P}}, e'_t\in P', [\mathcal{N}^{hom}_P\oplus \mathcal{N}^{adj}_{\tilde P}]\cap \g_{ti}^{{P}\tilde{P}}, e"_t\in P",$
has a unique solution and restriction 
of the solution to $P'\uplus P"$ gives a vector $(v^{tv}_{P'}, i^{tv}_{P"})$ such that
$(v^{tv}_{P'}, -i^{tv}_{P"})\in \breve{\V}_{P'P"}$
or a vector $(v^{ti}_{P'}, i^{ti}_{P"})$ such that
$(v^{ti}_{P'}, -i^{ti}_{P"})\in \breve{\V}_{P'P"}.$

\item The vectors $(v^{tv}_{P'}, -i^{tv}_{P"}),\  t=1, \cdots , |P|,
(v^{ti}_{P'}, -i^{ti}_{P"}),\  t=1, \cdots , |P|,$ 
form a generating set for $\breve{\V}_{P'P"}.$
\end{enumerate}
\end{lemma}
\begin{proof}
We only prove parts 2 and 3 since part 1 has already been shown.\\
2. If we replace the device characteristic 
of $ [\mathcal{N}_P\oplus \mathcal{N}^{adj}_{\tilde P}]\cap \g^{{P}\tilde{P}},$
or that of  $[\mathcal{N}_P\oplus \mathcal{N}^{adj}_{\tilde P}]\cap \g_{tv}^{{P}\tilde{P}}$ or that of  $[\mathcal{N}^{hom}_P\oplus \mathcal{N}^{adj}_{\tilde P}]\cap \g_{ti}^{{P}\tilde{P}}$ by its vector space translate, i.e., by
${\V}_{S'S"}\oplus {\V}^{adj}_{S'S"} \oplus \g^{{P}\tilde{P}}$ 
we get the device characteristic of $[\mathcal{N}^{hom}_P\oplus \mathcal{N}^{adj}_{\tilde P}]\cap \g^{{P}\tilde{P}}$ which is also on the same graph 
$\G_{SP}\oplus \G_{\tilde{S}\tilde{P}}$ as are $ [\mathcal{N}_P\oplus \mathcal{N}^{adj}_{\tilde P}]\cap \g^{{P}\tilde{P}}$ as well as $[\mathcal{N}^{hom}_P\oplus \mathcal{N}^{adj}_{\tilde P}]\cap \g_{tv}^{{P}\tilde{P}}$ and $[\mathcal{N}^{hom}_P\oplus \mathcal{N}^{adj}_{\tilde P}]\cap \g_{ti}^{{P}\tilde{P}}.$\\ 
By Corollary \ref{cor:regularmultiportconditions0},
 a linear network with a proper device characteristic, has a unique solution iff 
another obtained by replacing its affine device characteristic by
the vector space translate of the latter, has a
unique solution. We know that $ [\mathcal{N}_P\oplus \mathcal{N}^{adj}_{\tilde P}]\cap \g^{{P}\tilde{P}}$ has a proper device characteristic and has a unique solution. Thus $[\mathcal{N}^{hom}_P\oplus \mathcal{N}^{adj}_{\tilde P}]\cap \g^{{P}\tilde{P}},$ 
has a unique solution and therefore also $[\mathcal{N}^{hom}_P\oplus \mathcal{N}^{adj}_{\tilde P}]\cap \g_{tv}^{{P}\tilde{P}}$ and $[\mathcal{N}^{hom}_P\oplus \mathcal{N}^{adj}_{\tilde P}]\cap \g_{ti}^{{P}\tilde{P}}.$ \\
The restriction of a solution $(v_{S'},v_{P'}, v_{\tilde{S}'}, v_{\tilde{P}'},
i_{S"},i_{P"}, i_{\tilde{S}"}, i_{\tilde{P}"}) $ of $[\mathcal{N}^{hom}_P\oplus \mathcal{N}^{adj}_{\tilde P}]\cap \g_{tv}^{{P}\tilde{P}}$ or $[\mathcal{N}^{hom}_P\oplus \mathcal{N}^{adj}_{\tilde P}]\cap \g_{ti}^{{P}\tilde{P}}$
to $S'\uplus P'\uplus S"\uplus P"$ gives a solution of $\mathcal{N}^{hom}_P.$
Its restriction to $ P'\uplus P"$ gives $(v_{P'},i_{P"}).$  By Theorem \ref{thm:translatemultiport}, $(v_{P'},-i_{P"})\in \breve{\V}_{P'P"}.$
\\
3. Let $\breve{\V}_{P'P"}$ be the solution space of 
$Bv_{P'} -Qi_{P"} =0$ 
and let $(\breve{\V}^{adj}_{P'P"})_{\tilde{P}'\tilde{P}"}$
be the solution space of 
$Q^{\perp}v_{\tilde{P}'}+B^{\perp}v_{\tilde{P}"}=0.$
A vector being the restriction of a solution of $ [\mathcal{N}_P\oplus \mathcal{N}^{adj}_{\tilde P}]\cap \g_{tv}^{{P}\tilde{P}}$ to $P'\uplus P"\uplus \tilde{P}'\uplus \tilde{P}"$ 
is equivalent to its being a solution to the equation
\begin{align}
\label{eqn:primaldual1}
\ppmatrix{
        B  &  -Q & 0 & 0\\
        0 &  0  & Q^{\perp} & B^{\perp}\\
I &  0  & 0 & I\\
0&  I&-I& 0}
\ppmatrix{v_{P'}\\i_{P"}\\v_{\tilde{P}'}\\i_{\tilde{P}"}}&=\ppmatrix{0\\ 0\\-I^t\\0
} 
\end{align}
and a vector being the restriction of a solution of $ [\mathcal{N}_P\oplus \mathcal{N}^{adj}_{\tilde P}]\cap \g_{ti}^{{P}\tilde{P}}$ to $P'\uplus P"\uplus \tilde{P}'\uplus \tilde{P}"$
is equivalent to its being a solution to the equation
\begin{align}
\label{eqn:primaldual2}
\ppmatrix{
        B  &  -Q & 0 & 0\\
        0 &  0  & Q^{\perp} & B^{\perp}\\
I &  0  & 0 & I\\
0&  I&-I& 0}
\ppmatrix{v_{P'}\\i_{P"}\\v_{\tilde{P}'}\\i_{\tilde{P}"}}&=\ppmatrix{0\\ 0\\0\\-I^t
}, 
\end{align}
where the row spaces of $(B|-Q), (B^{\perp}|-Q^{\perp})$ are complementary
orthogonal and $I^t$ denotes the $t^{th}$ column of a  $|P|\times |P|$ identity matrix.
In the variables $v_{P'},i_{P"},$ Equation \ref{eqn:primaldual1} reduces to

\begin{align}
\label{eqn:primaldual3}
\ppmatrix{
        B  &  -Q \\
        B^{\perp} & -Q^{\perp}}
\ppmatrix{v_{P'}\\i_{P"}}&=\ppmatrix{0\\-B^{\perp}I^t
},
\end{align}
and Equation \ref{eqn:primaldual2} reduces to

\begin{align}
\label{eqn:primaldual4}
\ppmatrix{
        B  &  -Q \\
        B^{\perp} & -Q^{\perp}}
\ppmatrix{v_{P'}\\i_{P"}}&=\ppmatrix{0\\Q^{\perp}I^t
}.
\end{align}
It is clear that a vector belongs to $\breve{\V}_{P'P"}$ iff it is a 
solution of 
\begin{align}
\label{eqn:primaldual5}
\ppmatrix{
        B  &  -Q \\
        B^{\perp} & -Q^{\perp}}
\ppmatrix{v_{P'}\\i_{P"}}&=\ppmatrix{0\\x
},
\end{align}
for some  vector $x.$
The space of all such $x$ vectors is the column space of the matrix 
$(B^{\perp} |-Q^{\perp}).$ Noting that, for any matrix $K,$ the product  $KI^t$ is the $t^{th}$ column 
of $K,$ we see that the solutions, for $t=1, \cdots , |P|,$   
of Equations \ref{eqn:primaldual3}
and \ref{eqn:primaldual4}, span $\breve{\V}_{P'P"}.$

\end{proof}
 
\section{Maximum power transfer for linear multiports}
\label{sec:maxpower}
The maximum power transfer theorem, as originally stated, says that 
a linear resistive $1$-port transfers maximum power to a resistive load if 
the latter has value equal to the Thevenin resistance of the $1$-port.
In the multiport case the Thevenin equivalent is a resistor matrix 
whose transpose has to be connected to the multiport for maximum power transfer.
It was recognized early that a convenient way of studying maximum power
transfer is to study the port conditions for which such a transfer occurs
(\cite{desoer1,narayananmp}). We will use this technique 
to obtain such port conditions for an  affine multiport behaviour.
In the general, not necessarily strictly passive, case,  we can only try to obtain
stationarity of power transfer, rather than maximum power transfer.
After obtaining these conditions we show that they are, in fact,
 achieved if at all, when the 
multiport is terminated by its adjoint through an ideal transformer.
This means that the multiport behaviour 
need only be available as the port behaviour of a multiport $\N_P,$
and not explicitly, as an affine space $\breve{\A}_{P'P"}.$ 

\subsection{Stationarity of power transfer for linear multiports}
\label{subsec:maxpower}
Our convention for the sign of power associated with a multiport 
behaviour is
that when \\$(\breve{v}_{P'},\breve{i}_{P"})\in \breve{\A}_{P'P"},$
the  {\bf power absorbed} by the multiport behaviour  $\breve{\A}_{P'P"}$
is $\breve{v}_{P'}^T\breve{i}_{P"}.$ The power delivered by it is therefore $-\breve{v}_{P'}^T\breve{i}_{P"}.$

Suppose the multiport behaviour $\breve{\A}_{P'P"}$ is the solution space
of the equation $B_{P'}\breve{v}_{P'}-Q_{P"}\breve{i}_{P"}=s,$ with linearly independent
rows so that its  vector space translate $\breve{\V}_{P'P"}$
is the solution space of the equation $B_{P'}\breve{v}_{P'}-Q_{P"}\breve{i}_{P"}=0.$
We will now derive stationarity conditions on $(\breve{v}_{P'},\breve{i}_{P'})$ for the absorbed power $\breve{v}_{P'}^T\breve{i}_{P"}$ (i.e., { delivered power} $-\breve{v}_{P'}^T\breve{i}_{P"}$). 
If the behaviour does not satisfy the conditions, it would mean
that it has no stationary vectors and therefore there is no port condition
at which maximum power transfer (i.e., delivery) occurs.
Even if there are stationary vectors they only correspond to local minimum
or maximum power transfer.
\\
We have the optimization problem
\begin{align}
\label{eqn:optprob}
\mbox{minimize} \ \ \breve{v}_{P'}^T\breve{i}_{P"}\\
\label{eqn:feasible}
B_{P'}\breve{v}_{P'}-Q_{P"}\breve{i}_{P"}=s.
\end{align}

Let $(\delta \breve{v}_{P'}, \delta \breve{i}_{P"})$ be a perturbation,
consistent with Equation \ref{eqn:feasible},
about $(\breve{v}^{stat}_{P'},\breve{i}^{stat}_{P"})$ which we will take to be a stationary point of the optimization problem \ref{eqn:optprob}.
We then have

\begin{align}
\label{eqn:optprob3}
 (\delta \breve{v}_{P'})^T\breve{i}^{stat}_{P"}+ (\breve{v}^{stat}_{P'})^T(\delta \breve{i}_{P"})=0 \mbox{\ whenever}\\
(B_{P'}(\delta\breve{v}_{P'})-Q_{P"}(\delta \breve{i}_{P"}))=0.
\end{align}
Therefore, for some vector $\lambda ,$ we have 
\begin{align}
\label{eqn:optprob4}
 ((\breve{i}^{stat}_{P"})^T| (\breve{v}^{stat}_{P'})^T)-\lambda^T(B_{P'}|-Q_{P"})
=0.,
\end{align}
\\
Equivalently, 
the vector $((\breve{v}^{stat}_{P'})^T, (\breve{i}^{stat}_{P"})^T)$ must belong to the row space of
$(-Q_{P'}|B_{P"}).$
\\
Since we must have 
$$B_{P'}\breve{v}^{stat}_{P'}-Q_{P"}\breve{i}^{stat}_{P"}=s,$$
we get the condition  for stationarity,
\begin{align}
\label{eqn:optprob5}
\ppmatrix{B_{P'}&-Q_{P"}}\ppmatrix{-Q_{P'}^T\\B_{P"}^T}\lambda=s.
\end{align}

We note that, even when the multiport is regular, 
%
%
%
%
%
the coefficient matrix of Equation \ref{eqn:optprob5} may be singular
and the equation may have no solution, in which case we have no stationary vectors
for power transfer. If the coefficient matrix is nonsingular,
Equation \ref{eqn:optprob5} has a unique solution and using that $\lambda $ vector
we get a unique stationary vector $(\breve{v}^{stat}_{P'}, \breve{i}^{stat}_{P"}).$\\
The vector space translate, $\breve{\V}_{P'P"}$ of $\breve{\A}_{P'P"}$ is the solution space 
of the equation, $B_{P'}\breve{v}_{P'}-Q_{P"}\breve{i}_{P"}=0,$\\
  and $ (\breve{\V}_{P'P"}^{adj})_{P'(-P")}\equivd (\breve{\V}^{\perp}_{P'P"})_{P"P'},$ is the row space of $(-Q_{P'}|B_{P"}).$\\
Thus, the stationarity condition says that $(\breve{v}^{stat}_{P'},\breve{i}^{stat}_{P"}),$ 
belongs to $\breve{\A}_{P'P"}\cap(\breve{\V}_{P'P"}^{adj})_{P'(-P")}.$\\
We next show that the stationarity condition is achieved at the ports of the multiport $\N_P,$
if we terminate it by $\N^{adj}_{\tilde{P}}$ through the ideal transformer $\T^{P\tilde{P}}$ (resulting in the network $[\N_P\oplus\N^{adj}_{\tilde{P}}]\cap \T^{P\tilde{P}}$).

\begin{theorem}
\label{thm:maxpowerport}
Let  $\mathcal{N}_P,$ on graph $\G_{SP}$
and device characteristic ${\A}_{S'S"},$ have the port behaviour
$\breve{\A}_{P'P"}.$ 
Let ${\V}_{S'S"}, \breve{\V}_{P'P"}$ be the vector space translates of ${\A}_{S'S"}, \breve{\A}_{P'P"},$
respectively.
Let $\mathcal{N}^{adj}_{\tilde P}$ be on the disjoint copy $\G_{\tilde{S}\tilde{P}}$ of $\G_{SP},$ with device characteristic ${\V}^{adj}_{\tilde{S}'\tilde{S}"}\equivd ({\V}^{adj}_{S'S"})_{\tilde{S}'\tilde{S}"}.$ 
\begin{enumerate}
\item A vector $(v_{P'},i_{P"})$ is the restriction of a solution of the network $\N^{large}\equivd [\mathcal{N}_P\oplus \mathcal{N}^{adj}_{\tilde P}]\cap \T^{{P}\tilde{P}}$ to $P'\uplus P",$ iff 
$(v_{P'},-i_{P"})\in \breve{\A}_{P'P"}\cap(\breve{\V}_{P'P"}^{adj})_{P'(-P")}.$
\item  
%
%
Let $(\breve{v}^{stat}_{P'},\breve{i}^{stat}_{P"})\in
\breve{\A}_{P'P"}.$ Then $(\breve{v}^{stat}_{P'},\breve{i}^{stat}_{P"})$ satisfies 
stationarity condition with respect to $\breve{v}_{P'}^T\breve{i}_{P"},$
$(\breve{v}_{P'},\breve{i}_{P"})\in
\breve{\A}_{P'P"}$ 
 iff $(\breve{v}^{stat}_{P'},\breve{i}^{stat}_{P"})\in
\breve{\A}_{P'P"}\cap(\breve{\V}_{P'P"}^{adj})_{P'(-P")}.$
\item Let $(v^{1}_{P'},-i^{1}_{P"})$ be the restriction of a solution of the multiport $\N_P,$ to $P'\uplus P".$
Then $(\breve{v}^{1}_{P'},\breve{i}^{1}_{P"})$
satisfies the stationarity condition with respect to $\breve{v}_{P'}^T\breve{i}_{P"},$
$(\breve{v}_{P'},\breve{i}_{P"})\in
\breve{\A}_{P'P"}$ iff $(v^{1}_{P'},-i^{1}_{P"})$  is the restriction of a solution of the network
$[\mathcal{N}_P\oplus \mathcal{N}^{adj}_{\tilde P}]\cap \T^{{P}\tilde{P}},$
to $P'\uplus P".$

\end{enumerate}
\end{theorem}
\begin{proof}
1. The restriction of the set of solutions of  $\mathcal{N}_P$ on graph $\G_{SP}$ to $P'\uplus P",$
is\\ 
$[(\V_{S'P'}\oplus (\V^{\perp}_{S'P'})_{S"P"})\cap \A_{S'S"}]\circ {P'P"}.$
This is  the same as
$(\breve{\A}_{P'P"})_{P'(-P")}.$\\
The restriction of the set of solutions of $\mathcal{N}^{adj}_{\tilde P}$ on the disjoint copy $\G_{\tilde{S}\tilde{P}}$ of $\G_{SP},$ to $\tilde{P}'\uplus \tilde{P}"$ is\\
$[(\V_{\tilde{S}'\tilde{P}'}\oplus (\V^{\perp}_{\tilde{S}'\tilde{P}'})_{\tilde{S}"\tilde{P}"})\cap {\V}^{adj}_{\tilde{S}'\tilde{S}"}]\circ {\tilde{P}'\tilde{P}"}.$ 
This we know, by Theorem \ref{thm:adjointmultiport},
 to be the same as
$(\breve{\V}^{adj}_{P'P"})_{\tilde{P}'-\tilde{P}"}.$\\
Let $\mathcal{N}^{adj}_{P}$ be on  $\G_{SP},$ with device characteristic ${\V}^{adj}_{{S}'{S}"}.$\\
The restriction of the set of solutions of $\mathcal{N}^{adj}_{P}$ to ${P}'\uplus {P}"$ is 
$[(\V_{{S}'{P}'}\oplus (\V^{\perp}_{{S}'{P}'})_{{S}"{P}"})\cap {\V}^{adj}_{{S}'{S}"}]\circ {{P}'{P}"}.$\\
Clearly,
$[((\V_{\tilde{S}'\tilde{P}'}\oplus (\V^{\perp}_{\tilde{S}'\tilde{P}'})_{\tilde{S}"\tilde{P}"})\cap {\V}^{adj}_{\tilde{S}'\tilde{S}"})\circ {\tilde{P}'\tilde{P}"}]_{P'P"}=[(\V_{{S}'{P}'}\oplus (\V^{\perp}_{{S}'{P}'})_{{S}"{P}"})\cap {\V}^{adj}_{{S}'{S}"}]\circ {{P}'{P}"}=(\breve{\V}^{adj}_{P'P"})_{P'(-P")}$
i.e., $[(((\V_{\tilde{S}'\tilde{P}'}\oplus (\V^{\perp}_{\tilde{S}'\tilde{P}'})_{\tilde{S}"\tilde{P}"})\cap {\V}^{adj}_{\tilde{S}'\tilde{S}"})\cap \T^{P\tilde{P}})\circ {P'P"}
=([(\V_{{S}'{P}'}\oplus (\V^{\perp}_{{S}'{P}'})_{{S}"{P}"})\cap {\V}^{adj}_{{S}'{S}"}]\circ {{P}'{P}"})_{P'(-P")}$\\$= \breve{\V}^{adj}_{P'P"},$
since vectors in $ \T^{P\tilde{P}}$ are precisely the ones that satisfy $v_{P'}=v_{\tilde{P}"}, i_{P"}=-i_{\tilde{P}"}
.$\\
The restriction of the set of solutions of $[\mathcal{N}_P\oplus \mathcal{N}^{adj}_{\tilde P}]\cap \T^{{P}\tilde{P}},$
to $P'\uplus P"$ is therefore equal to \\
$[[(\V_{S'P'}\oplus (\V^{\perp}_{S'P'})_{S"P"})\cap \A_{S'S"}]\circ {P'P"}]
\cap [[(\V_{\tilde{S}'\tilde{P}'}\oplus (\V^{\perp}_{\tilde{S}'\tilde{P}'})_{\tilde{S}"\tilde{P}"})\cap {\V}^{adj}_{\tilde{S}'\tilde{S}"}\cap \T^{P\tilde{P}}]\circ {{P}'{P}"}],$
\\
$=(\breve{\A}_{P'P"})_{P'(-P")}\cap\breve{\V}_{P'P"}^{adj}.$
Thus a vector $(v_{P'},i_{P"})$ is the restriction of a solution of the network\\ $\N^{large}\equivd [\mathcal{N}_P\oplus \mathcal{N}^{adj}_{\tilde P}]\cap \T^{{P}\tilde{P}}$ to $P'\uplus P",$ iff
$(v_{P'},-i_{P"})\in \breve{\A}_{P'P"}\cap(\breve{\V}_{P'P"}^{adj})_{P'(-P")}.$
\\
Parts 2 and 3 follow from part 1 and the discussion preceding the theorem.
\end{proof}
\subsection{Maximum Power Transfer Theorem for passive multiports}
\label{subsec:mptpassive}
We show       below that the stationarity conditions of the previous subsection reduce       to maximum power transfer conditions when the multiport is passive.
\\
A vector space $\V_{S'S"}$ is \nw{passive}, iff $\langle x_{S'},y_{S"} \rangle\geq 0,$ whenever $(x_{S'},y_{S"})\in \V_{S'S"}.$\\ 
It is \nw{strictly passive} iff $\langle x_{S'},y_{S"} \rangle > 0,$ whenever $(x_{S'},y_{S"})\in \V_{S'S"}, (x_{S'},y_{S"})\ne 0_{S'S"}.$\\
An affine space   $\A_{S'S"}$  is (strictly) passive iff its vector space translate is (strictly) passive. 
A multiport is (strictly) passive iff its device characteristic 
is (strictly) passive.

We now have a routine result which links passivity of the device characteristic
of a multiport to its port behaviour.

\begin{lemma}
\label{lem:maxpower}
Let $\N_P$ be a multiport on graph $\G_{SP}$ with device characteristic 
$\A_{S'S"}.$ 
\begin{enumerate}
\item If $\A_{S'S"}$ is passive so is the port behaviour $\breve{\A}_{P'P"}$ of $\N_P.$
\item If $\A_{S'S"}$ is strictly passive and $P$ contains no loops or cutsets of $\G_{SP},$ then the port behaviour $\breve{\A}_{P'P"}$ of $\N_P$
 is also strictly passive.
\end{enumerate}
\end{lemma}
\begin{proof}
1. We assume that $\breve{\A}_{P'P"}$ is nonvoid. We have,\\
$\breve{\A}_{P'P"}= ((\V_{S'P'}\oplus (\V^{\perp}_{S'P'})_{S"P"})\lrar \A_{S'S"})_{P'(-P")} ,$ where $\V_{S'P'}\equivd (\V^v(\G_{SP}))_{S'P'}.$\\
By Theorem \ref{thm:IIT2}, it follows that its vector space translate \\
$\breve{\V}_{P'P"}= ((\V_{S'P'}\oplus (\V^{\perp}_{S'P'})_{S"P"})\lrar \V_{S'S"})_{P'(-P")} .$\\ 
%
%
%
Let  $(v_{P'},-i_{P"})$ belong to $ \breve{\V}_{P'P"}.$
Then there exist $(v_{S'},v_{P'})\in \V_{S'P'}$ and $(i_{S"},i_{P"})\in (\V^{\perp}_{S'P'})_{S"P"},$
such that $(v_{S'},i_{S"})\in \V_{S'S"}.$\\
By the orthogonality of $\V_{S'P'},\V^{\perp}_{S'P'},$ it follows that
$\langle (v_{S'},v_{P'}),(i_{S"},i_{P"})\rangle =\langle v_{S'},i_{S"}\rangle+\langle v_{P'},i_{P"}\rangle=0,$ and \\
by the passivity of $ \V_{S'S"},$ it follows that $\langle v_{S'},i_{S"}\rangle \geq 0.$\\
Therefore  $\langle v_{P'},-i_{P"}\rangle \geq 0.$\\
2. Without loss of generality, we assume that the graph $\G_{SP}$ is connected.
If $P$  contains no cutset or circuit of $\G_{SP}, $ then $S$ contains both 
a tree as well as a cotree of $\G_{SP}.$
If the voltages assigned to the branches of a tree are zero, the branches 
in its complement will have zero voltage. Therefore $\V_{S'P'}\times P'$
must necessarily be a zero vector space.
If the currents in the branches of a cotree are zero the branches
in its complement will have zero current. Therefore $\V^{\perp}_{S'P'}\times P'$
must necessarily be a zero vector space.\\
Now let $(v_{P'},-i_{P"})\in \breve{\V}_{P'P"}, (v_{P'},-i_{P"})\ne 0_{P'P"}.$
As in part 1 above, there exist $(v_{S'},v_{P'})\in \V_{S'P'}$ and $(i_{S"},i_{P"})\in \V^{\perp}_{S'P'},$
such that $(v_{S'},i_{S"})\in \V_{S'S"}.$\\
Since $\V_{S'P'}\times P'$
and $\V^{\perp}_{S'P'}\times P'$ are both zero vector spaces,
we must have that $(v_{S'},i_{S"})\ne 0_{S'S"}$  so that,
by strict passivity of $\V_{S'S"},$ we have  
$\langle v_{S'},i_{S"}\rangle > 0.$
Further $\langle (v_{S'},v_{P'}),(i_{S"},i_{P"})\rangle =0,$ so that
we can conclude $\langle v_{P'},-i_{P"}\rangle > 0.$
\end{proof}

When a port behaviour is passive or strictly passive, by taking into account second order terms, 
we can show that  the stationarity condition implies a  maximum power delivery condition.

Let $(\breve{v}^{stat}_{P'},\breve{i}^{stat}_{P"})$ satisfy the stationarity condition 
$((\breve{v}^{stat}_{P'})^T,(\breve{i}^{stat}_{P"})^T) = \hat{\lambda}^T(-Q_{P'}|B_{P"}),$ for some $\hat{\lambda}.$\\ 
Now  for any $(\breve{v}_{P'}, \breve{i}_{P"})\in \breve{\A}_{P'P"},$ 
writing $(\breve{v}_{P'}, \breve{i}_{P"})=(\breve{v}^{stat}_{P'}+\Delta\breve{v}_{P'},\breve{i}^{stat}_{P"}+\Delta\breve{i}_{P"}),$
we have\\
$\langle \breve{v}_{P'}, \breve{i}_{P"}\rangle =\langle \breve{v}_{P'}, \breve{i}_{P"}\rangle - \hat{\lambda}^T[(B_{P}\breve{v}_{P'}-Q_{P"}\breve{i}_{P"})-s]$
\\$= 
\langle\breve{v}^{stat}_{P'},\breve{i}^{stat}_{P"}\rangle+ (\Delta \breve{v}_{P'})^T\breve{i}^{stat}_{P"}+ (\breve{v}^{stat}_{P'})^T\Delta \breve{i}_{P"}+\langle \Delta \breve{v}_{P'}, \Delta \breve{i}_{P"}\rangle-\hat{\lambda}^T(B_{P'}\Delta(\breve{v}_{P'})-Q_{P"}(\Delta \breve{i}_{P"})). $
\\ We can rewrite the right side as  \\ 
$\langle\breve{v}^{stat}_{P'},\breve{i}^{stat}_{P"}\rangle+ \langle \Delta \breve{v}_{P'}, \Delta \breve{i}_{P"}\rangle +((\breve{v}^{stat}_{P'})^T+\hat{\lambda}^TQ_{P"})\Delta\breve{i}_{P"}+((\breve{i}^{stat}_{P"})^T-{\lambda}^TB_{P'})\Delta \breve{v}_{P'}. $\\
Applying the condition 
$((\breve{v}^{stat}_{P'})^T,(\breve{i}^{stat}_{P"})^T) = \hat{\lambda}^T(-Q_{P'}|B_{P"}),$ 
this expression reduces to\\
$\langle\breve{v}^{stat}_{P'},\breve{i}^{stat}_{P"})+\langle \Delta \breve{v}_{P'}, \Delta \breve{i}_{P"}\rangle .$
Therefore, $\langle \breve{v}_{P'}, \breve{i}_{P"}\rangle =
\langle\breve{v}^{stat}_{P'},\breve{i}^{stat}_{P"})+\langle \Delta \breve{v}_{P'}, \Delta \breve{i}_{P"}\rangle .$

If $\breve{\A}_{P'P"}$ is passive we have $\langle \Delta \breve{v}_{P'}, \Delta \breve{i}_{P"}\rangle\geq 0 ,$
so that 
$\langle \breve{v}_{P'}, \breve{i}_{P"}\rangle \geq 
\langle\breve{v}^{stat}_{P'},\breve{i}^{stat}_{P"}\rangle.$\\
Equivalently, the power delivered $(-\breve{v}_{P'}^T \breve{i}_{P"}),$ by 
$\breve{\A}_{P'P"},$ is maximum when $(\breve{v}_{P'}, \breve{i}_{P"})= (\breve{v}^{stat}_{P'},\breve{i}^{stat}_{P"}).$
\\If $\breve{\A}_{P'P"}$ is strictly passive we have
$\langle \breve{v}_{P'}, \breve{i}_{P"}\rangle  >
\langle\breve{v}^{stat}_{P'},\breve{i}^{stat}_{P"}\rangle$
whenever $(\Delta \breve{v}_{P'}, \Delta \breve{i}_{P"})\ne 0,$\\
so that 
$(\breve{v}^{stat}_{P'},\breve{i}^{stat}_{P"})$
is the unique maximum delivery vector in $\breve{\A}_{P'P"}.$

By Lemma \ref{lem:maxpower}, if a multiport $\N_P$ is passive,
so is its port behaviour. Therefore, Equation \ref{eqn:optprob5}
also gives the condition for maximum power transfer from a passive $\N_P.$

We thus have, from the above discussion, using Theorem \ref{thm:maxpowerport},
the following result.
%
%
%
%
\begin{theorem}
\label{thm:passivemaxpowerport}
Let  $\mathcal{N}_P,$ on graph $\G_{SP}$
and with passive device characteristic ${\A}_{S'S"},$ have the port behaviour
$\breve{\A}_{P'P"}.$
Let ${\V}_{S'S"}$ be the vector space translate of ${\A}_{S'S"}.$
Let $\mathcal{N}^{adj}_{\tilde P}$ be on the disjoint copy $\G_{\tilde{S}\tilde{P}}$ of $\G_{SP},$ with device characteristic ${\V}^{adj}_{\tilde{S}'\tilde{S}"}.$
\begin{enumerate}
\item
Let $(\breve{v}^{stat}_{P'},\breve{i}^{stat}_{P"})\in
\breve{\A}_{P'P"}.$ Then $(\breve{v}^{stat}_{P'},\breve{i}^{stat}_{P"})$ satisfies
\begin{align}
\label{eqn:maximize}
 maximize \ \ (-\breve{v}_{P'}^T\breve{i}_{P"}),\ \ \ \ \ \
(\breve{v}_{P'},\breve{i}_{P"})\in
\breve{\A}_{P'P"},
\end{align}
 iff $(\breve{v}^{stat}_{P'},\breve{i}^{stat}_{P"})\in
\breve{\A}_{P'P"}\cap(\breve{\V}_{P'P"}^{adj})_{P'(-P")}.$
\item Let $(\breve{v}^{stat}_{P'},-\breve{i}^{stat}_{P"})$ be the restriction of a solution of the multiport $\N_P,$ to $P'\uplus P".$
Then $(\breve{v}^{stat}_{P'},\breve{i}^{stat}_{P"})$
satisfies the optimization condition in Equation \ref{eqn:maximize},
iff $(\breve{v}^{stat}_{P'},-\breve{i}^{stat}_{P"})$  is the restriction of a solution of the network
$[\mathcal{N}_P\oplus \mathcal{N}^{adj}_{\tilde P}]\cap \T^{{P}\tilde{P}},$
to $P'\uplus P".$
\end{enumerate}
\end{theorem}
\begin{example}
Here are some examples of passive (strictly passive) multiports.
\begin{enumerate}
\item $\breve{\V}_{P'P"}$ satisfies
 $v_{P'}=Ri_{P"},$ $R$ positive semidefinite matrix (positive definite matrix);
\item $\breve{\V}_{P'P"}$  is an ideal transformer. Here $\langle \breve{v}_{P'}, \breve{i}_{P"}\rangle =0,$
indeed $\breve{\V}_{P'P"} = \breve{\V}_{P'P"}\circ P'\oplus \breve{\V}_{P'P"}\circ P".$
 \item $\breve{\V}_{P'P"}$  is a gyrator. Here $\langle \breve{v}_{P'}, \breve{i}_{P"}\rangle =0,$
\end{enumerate}
\end{example}
\begin{remark}
1. Let  $\breve{\V}_{P'P"}$ be strictly passive.
If $(0_{P'},i_{P"})$ or $(v_{P'},0_{P"})$ belongs to $\breve{\V}_{P'P"},$ 
its strict passivity is violated. Therefore $r(\breve{\V}_{P'P"}\times P')$
as well as $r(\breve{\V}_{P'P"}\times P")$ must be zero and consequently
$r(\breve{\V}_{P'P"}\circ P")$
as well as $r(\breve{\V}_{P'P"}\circ P')$ must equal $|P|,$ using 
Theorem \ref{thm:dotcrossidentity}.
Let  $(C_{P'}|E_{P"})$ be the representative matrix of
 $\breve{\V}_{P'P"}.$ It is clear that $C_{P'},E_{P"}$ 
are nonsingular so that by invertible row transformation, we can reduce 
$(C_{P'}|E_{P"})$ to the form $(Z^T|I)$ or $(I|Y^T).$ We then have $i_{P"}^Tv_{P'}=
i_{P"}^TZi_{P"}.$ Thus strict passivity of $\breve{\V}_{P'P"},$
implies properness and 
is equivalent to positive definiteness of the resistive matrix $Z$
or that of the conductance  matrix $Y.$
\\
2. Suppose $\breve{\A}_{P'P"}$ is the solution space of the equation
$Iv_{P'}-Zi_{P"} = \mathcal{E},$ i.e., $v_{P'}=Zi_{P"} + \mathcal{E}.$  In this case the stationarity 
condition reduces to $(Z+Z^T)i_{P"}=\mathcal{E}.$ 
If $Z$ is positive definite, so would $(Z+Z^T)$ be and therefore also
be invertible. Therefore when $\breve{\A}_{P'P"}$  is strictly passive 
the maximum power transfer condition of Equation \ref{eqn:optprob5}  has a unique solution.
 But, in general, when $Z$ is not positive definite, $(Z+Z^T)$ 
  can be singular so that there is no solution
to the stationarity 
condition even if $Z$ is nonsingular.\\
3. Let $\breve{\V}_{P'P"}$ have  
a representative matrix of the form
$(I|K)$ or $(K|I),$ where  $K$ is nonsingular.
The adjoint $\breve{\V}^{adj}_{P'P"}$ also has a representative matrix
of the form
$(I|K^T)$ or $(K^T|I).$  
If $\breve{\V}_{P'P"}$ is strictly passive, we have that $K$ and 
therefore $K^T$ is positive definite.
Thus,
 the adjoint of a strictly passive multiport
behaviour is strictly passive.
 
\end{remark}
\section{Conclusions}
\label{sec:conclusions}
We have attempted to show that basic circuit theory can benefit in terms 
of clarity, rigour and efficiency of computation, if we use implicit linear algebra.

We have applied the implicit inversion theorem to linear multiport connection
and have used the implicit duality theorem to show that duality properties of device
characteristics of a linear multiport are inherited by the port behaviour.
\\We have introduced the notion of a regular multiport and the idea of termination by the adjoint through a gyrator or through  an ideal transformer.
We have used the former termination to 
give an algorithm for computing the port behaviour of such multiports using
easily available linear circuit simulators.
We have used the latter termination  to give a condition for 
maximum power transfer for general linear multiports.

Although most of the results are stated for linear multiports governed by 
linear equations over the reals, we have shown how to extend them 
when they are over the complex field in order to handle the steady state sinusoidal case.
\appendix
%
%
%
%
%
\section{Proof of a generalization of Theorem \ref{thm:IIT2}}
\label{sec:proofsiitidt}
We first prove a more general version of Theorem \ref{thm:IIT2} where,
in place  of the vector space $\Vsp,$ we have a collection of vectors
$\Ksp$ that is closed under subtraction. Note that such a collection contains
the zero vector and therefore has nonvoid contraction to any subset of $S\uplus P.$ We take the collections $ \Kpq,\Ksq$
to be arbitrary.
\begin{theorem}
\label{thm:IITA}
{\bf Implicit Inversion Theorem}
Let  $\Ksp$ be a collection of vectors   closed under subtraction and let $ \Kpq,\Ksq$
be collections of vectors  with $S,P,Q,$ being pairwise disjoint.
Consider the equation
\begin{align}
\label{eqn:IITA}
\Ksp \lrar \Kpq =\Ksq, 
\end{align}
with $\Kpq$ treated as unknown.
We have the following.
\begin{enumerate}
\item given $\Ksp, \Ksq,$  there exists $ \hat{\K}_{PQ}, $ satisfying  Equation \ref{eqn:IITA}, only if   $\Ksp\circ S\supseteq \Ksq \circ S$\\ and $\Ksq+\Ksp\times S \subseteq \Ksq.$
\item given $\Ksp, \Ksq,$ if  $\Ksp\circ S\supseteq \Ksq \circ S$ and 
$\Ksq+\Ksp\times S \subseteq \Ksq,$
then $\hat{\K}_{PQ}\equivd \Ksp \lrar \Ksq $
satisfies the equation.

\item given $\Ksp, \Ksq,$ assuming that  Equation \ref{eqn:IITA}, with $\Kpq$ treated as unknown, is satisfied by some $\hat{\K}_{PQ} ,$ it is unique if  the additional conditions $\Ksp\circ P\supseteq \Kpq \circ P$ and
$\Kpq+ \Ksp\times P\subseteq \Kpq$ are imposed.
\end{enumerate}
\end{theorem}

\begin{proof}
1.  $\ \ $ Suppose $\KSP \lrar \hat{\K}_{PQ} = \KSQ.$\\
If $(f_S,f_Q)\in \Ksq,$
there must exist $(f_S,f_P)\in \Ksp, (f_P,f_Q)\in \hat{\K}_{PQ},$ for some $f_P.$
Hence, $\KSP \circ S \supseteq \KSQ\circ S.$
Further, 
if $(g_S,0_P)\in \Ksp,$ then 
we must have $(f_S+g_S,f_P)\in \Ksp,$ since $\Ksp$ is closed under addition.
 Since $(f_P,f_Q)\in \hat{\K}_{PQ},$
it follows that
$(f_S+g_S,f_Q)\in \Ksq,$ i.e., $\Ksq\supseteq \Ksq+\Ksp\times S.$
\\
2. $\ \ $   
Let $\hat{\K}_{PQ} \equivd \KSP \lrar \KSQ,$ i.e., $\hat{\K}_{PQ}$ is the collection of all vectors $(f_P,f_Q) $  such that, for some
vector $f_S,$
$(f_S,f_Q)\in \KSQ$, $(f_S,f_P) \in \KSP.$
We will now show that $\KSP\lrar \hat{\K}_{PQ}= \KSQ.$\\
Let $(f_S,f_Q) \in \KSQ.$
Since $\KSP\circ S \supseteq \KSQ \circ S,$ we must have that $(f_S,f_P) \in \KSP,$ for some $f_P.$
By the definition of $\hat{\K}_{PQ},$ we have that $(f_P,f_Q) \in \hat{\K}_{PQ}.$
Hence,
$(f_S,f_Q) \in \KSP \lrar \hat{\K}_{PQ}.$
Thus,
$\KSP \lrar \hat{\K}_{PQ} \supseteq \KSQ.$
\\
Next, let $(f_S,f_Q)\in \KSP \lrar \hat{\K}_{PQ},$ i.e., for some $f_P, (f_S,f_P) \in \KSP$ and $(f_P,f_Q) \in \hat{\K}_{PQ}.$
\\
We know, by the definition of $\hat{\K}_{PQ},$ that there exists $f_S'$ 
such that $(f_S',f_Q) \in \KSQ$ and $(f_S', f_P) \in \KSP.$
\\
Since $\KSP$ is closed under subtraction, we must have, $(f_S - f_S',0_P) 
\in \KSP.$
Hence, $(f_S - f_S') \in \KSP \times S,$
\\
Since $\KSQ\supseteq \Ksq+\Ksp\times S,$ and
$(f_S',f_Q) \in \KSQ$,
it follows that $(f_S - f_S',  { 0}_{Q}) + (f_S',f_Q) = (f_S,f_Q)$ also
belongs to $\KSQ$.
Thus, $\KSP \lrar \hat{\K}_{PQ} \subseteq \KSQ.$
\\3. 
Let $\hat{\K}_{PQ}$ satisfy the equation  $\KSP\lrar \hat{\K}_{PQ} =\KSQ. $
\\
From the proof of part 2, we know that if $\hat{\K}_{PQ}$ satisfies $\KSP \circ P \supseteq \hat{\K}_{PQ}\circ P$ and  $\hat{\K}_{PQ} \supseteq \hat{\K}_{PQ}+\Ksp\times P,$ 
\\
then $\KSP\lrar (\KSP\lrar \hat{\K}_{PQ}) =\hat{\K}_{PQ}.$ But $\hat{\K}_{PQ}$ satisfies $\KSP\lrar \hat{\K}_{PQ} =\KSQ$ and satisfies $\KSP \circ P \supseteq \hat{\K}_{PQ}\circ P$ and  
$\hat{\K}_{PQ} \supseteq \hat{\K}_{PQ}+\Ksp\times P.$ 
It follows that for any such $\hat{\K}_{PQ},$ we have $\KSP\lrar \KSQ=\hat{\K}_{PQ}.$
\\
This proves that $\hat{\K}_{PQ}\equivd \KSP\lrar \KSQ$
is the only solution to the equation $\KSP\lrar \KPQ =\KSQ, $ under the condition $\KSP \circ P \supseteq \KPQ\circ P$ and  
$\Kpq \supseteq \Kpq+\Ksp\times P.$ 

\end{proof}
\section{Proof of a generalization of Theorem \ref{thm:idt0}}
Theorem \ref{thm:idt0}, is from \cite{HNarayanan1986a}.
We prove a general version of it below based on the proof given in \cite{HNarayananadjoint}.
It  generalizes naturally
to the matroid case \cite{STHN2014}.
Other proofs and applications may be found in
\cite{HNarayanan1997,schaft1999,narayanan2002some,HNarayanan2009}.
A version in the context of Pontryagin Duality is available in (\cite{forney2004}).

Let $\Omega_S$ denote a family of collections $\K_S$ of vectors on $S,$ 
$\Omega_{SP}$ denote a family of collections $\K_{SP}$ of vectors on $S\uplus P,$ $\Omega_{PQ}$ denote a family of collections $\K_{PQ}$ of vectors on $P\uplus Q.$

Let $\Omega \equivd \bigcup _{\{S,S\ finite\}}\Omega_S,$
Let $\Omega$ be closed under the operations on collections of vectors, of sum, intersection, contraction
and restriction (using the extended definition of sum and intersection of collections of vectors).
Let $d:\Omega\rightarrow \Omega$ be such that $d(\Omega_S)\subseteq \Omega_S.$
Further, let $d$ satisfy, for $\Ks\in \Omega_S, \K_W\in \Omega_W,$
\begin{enumerate}
\item $d(d(\K_S))=\K_S,$
\item $d(\Ks+\K_W) = d(\Ks)\cap d(\K_W),$
\item $d(\Ks \circ T)= d(\Ks)\times T,$
\item $d(\Ks \times T)= d(\Ks)\circ T.$
\end{enumerate}

In particular, 
\begin{enumerate}
\item 
$\Omega_S$ could denote the collection of all vector spaces 
over a field $\mathbb{F}$ on the set $S$ and $d(\K_S)$ could denote $\K_S^{\perp}.$
\item
$\Omega_S$ could denote the collection of all vector spaces
over  $\mathbb{C}$ on the set $S$ and $d(\K_S)$ could denote ${\K_S}^*,$
the collection of all vectors whose inner product with vectors in 
$\K_S$ is zero.
\item $\Omega_S$ could denote the collection of all finitely generated 
cones on $S$ over $\Re.$
In this case $d(\K_S)$ could denote $\K_S^{p},$
where $\K_S^{p}\equivd \{g_S,<f_S,g_S>\ \leq 0, f_S\in \K_S\}.$
Note that for a vector space $\Vs,$ $\V_S^{p}=\V_S^{\perp}.$ 
\end{enumerate}

\begin{theorem}
\label{thm:idtgen}
Let $\Ksp\in \Omega _{SP}, \Kpq \in \Omega _{PQ},$ 
with $S,P,Q,$ being pairwise disjoint.
\\
 We then have,
$d(\mathcal{K}_{SP}\leftrightarrow \mathcal{K}_{PQ}) 
\ \equaln\ d(\mathcal{V}_{SP}) \rightleftharpoons d(\mathcal{V}_{PQ}) 
.$ 
\end{theorem}

\begin{proof}
From the definition of matched  composition, 
\begin{align*}
(\mathcal{K}_{SP}\leftrightarrow \mathcal{K}_{PQ}) = (\mathcal{K}_{SP}+ \mathcal{K}_{(-P)Q}) \times (S\uplus Q)
\end{align*}
and also 
\begin{align*}
 (\mathcal{K}_{SP}\leftrightarrow \mathcal{K}_{PQ}) =(\mathcal{K}_{SP}\cap \mathcal{K}_{PQ})\circ (S\uplus Q).
\end{align*}
Similarly from the definition of skewed composition,
\begin{align*}
 (\mathcal{K}_{SP} \rightleftharpoons \mathcal{K}_{PQ}) = (\mathcal{K}_{SP}+ \mathcal{K}_{PQ}) \times (S\uplus Q)
\end{align*}
and also 
\begin{align*}
(\mathcal{K}_{SP} \rightleftharpoons \mathcal{K}_{PQ}) =  (\mathcal{K}_{SP}\cap \mathcal{K}_{(-P)Q})\circ (S\uplus Q).
\end{align*}
Hence we have
\begin{align*}
 d(\mathcal{K}_{SP}\leftrightarrow \mathcal{K}_{PQ}) 
& = d[(\mathcal{K}_{SP} + \mathcal{K}_{(-P)Q}) \times (S\uplus Q)] \\
& = [d(\mathcal{K}_{SP} + \mathcal{K}_{(-P)Q})] \circ (S\uplus Q) \\
& =[d(\mathcal{K}_{SP}) \cap d(\mathcal{K}_{(-P)Q})] \circ(S\uplus Q)\\
& = \ \ \  d(\mathcal{K}_{SP}) \rightleftharpoons d(\mathcal{K}_{PQ}). 
\end{align*}
\end{proof}
\section{Maximum Power Transfer Theorem for the complex case}
\label{subsec:mptcomplex}
\subsubsection{Stationarity}
We fix some notation to consider the complex field case (steady state sinusoidal analysis).
We take for any matrix $M,$ $\overline{M}$ to be the conjugate and
$M^*$ to be the conjugate transpose.
\\
In this case, the optimization
problem is
\begin{align}
\label{eqn:optprobcomplex}
\mbox{minimize} \ \ \breve{v}_{P'}^*\breve{i}_{P"}+\breve{i}_{P'}^*\breve{v}_{P"}(\equivd \overline{\breve{v}_{P'}^T}\breve{i}_{P"}+\overline{\breve{i}_{P'}^T}\breve{v}_{P"})\\
\label{eqn:feasible}
B_{P'}\breve{v}_{P'}-Q_{P"}\breve{i}_{P"}=s.
\end{align}


If $(\breve{v}^{stat}_{P'},\breve{i}^{stat}_{P"}),$
is a stationary point for the optimization problem \ref{eqn:optprobcomplex}
, we
have
\begin{align}
\label{eqn:optprob31}
 (\breve{i}^{stat}_{P"})^T\overline{\delta\breve{v}}_{P'}+ (\breve{v}^{stat}_{P'})^T\overline{\delta\breve{i}}_{P"}
+ (\breve{i}^{stat}_{P"})^*{\delta\breve{v}}_{P'}+ (\breve{v}^{stat}_{P'})^*{\delta\breve{i}}_{P"}
&=0,&
\end{align}
for every  vector $(\delta\breve{v}_{P'},\delta \breve{i}_{P"}),$
such that $({B}_{P'}({\delta\breve{v}}_{P'})-{{Q}}_{P"}({\delta \breve{i}}_{P"}))=0.$
Therefore we have, for some vector $\lambda,$
 $${((\breve{i}^{stat}_{P"})^T|(\breve{v}^{stat}_{P'})^T)-\lambda^T(\overline{B}_{P'}|-\overline{{Q}}_{P"})}
=0.$$
Thus the stationarity at $(\breve{v}^{stat}_{P'},\breve{i}^{stat}_{P"}),$ is equivalent to
$${((\breve{v}^{stat}_{P'})^T|(\breve{i}^{stat}_{P"})^T)-\lambda^T(-\overline{Q}_{P'}|\overline{{B}}_{P"})}
=0.$$
Since we must have
$$B_{P'}\breve{v}^{stat}_{P'}-Q_{P"}\breve{i}^{stat}_{P"}=s,$$
the stationarity condition reduces to
\begin{align}
\label{eqn:optprob51}
\ppmatrix{B_{P'}&-Q_{P"}}\ppmatrix{-Q_{P'}^*\\B_{P"}^*}\lambda=s,
\end{align}
By $\V^*_{S'S"},$ we mean the collection of all vectors whose inner product
with every vector in $\V_{S'S"}$ is zero.
To obtain the complex version of the statement of Theorem \ref{thm:maxpowerport} and its proof,
$\V^{\perp}_{S'S"}$ should be interpreted as $\V^*_{S'S"},$
 $\V^{adj}_{S'S"}$ should be taken to be
$ (\V^*_{S'S"})_{(-S")S'}.$
\subsection{Passive multiports}
We need to modify the definition of passivity in the complex case.
By $\langle x_{S'},y_{S"} \rangle,$ we mean the inner product
$y^*_{S'}x_{S'}.$
A vector space $\V_{S'S"}$ is \nw{passive}, iff $\langle x_{S'},y_{S"} \rangle
+\langle y_{S"},x_{S'} \rangle\geq 0,$ whenever $(x_{S'},y_{S"})\in \V_{S'S"}.$\\
It is \nw{strictly passive} iff $\langle x_{S'},y_{S"} \rangle
+\langle y_{S"},x_{S'} \rangle > 0,$ whenever $(x_{S'},y_{S"})\in \V_{S'S"}, (x_{S'},y_{S"})\ne 0_{S'S"}.$\\
An affine space   $\K_{S'S"}$  is (strictly) passive iff its vector space translate is (strictly) passive.

With this extended definition of passivity,  taking $\V^{\perp}_{S'S"}$ to be $\V^*_{S'S"},$
 $\V^{adj}_{S'S"}$  to be
$ (\V^*_{S'S"})_{(-S")S'},$ and the optimization problem to be
maximization of $\langle -x_{P'},y_{P"} \rangle
+\langle -y_{P"},x_{P'} \rangle > 0,$ whenever $(x_{P'},y_{P"})\in \breve{\A}_{P'P"},$
the statement and proof of Lemma \ref{lem:maxpower} and Theorem \ref{thm:passivemaxpowerport} go through.

\bibliographystyle{elsarticle1-num}
\bibliography{references}

\end{document}